\documentclass{sig-alternate}
\usepackage{dsfont}
\usepackage{subfigure}
\usepackage{graphics}
\newcommand{\etal}{{\em et al. }}
\newcommand{\norm}[1]{{\left\Vert#1\right\Vert}_2}
\newcommand{\prob}[1]{{\bf \mbox{\bf Pr}}\left[#1\right]}
\newcommand{\probomega}[1]{{\bf\mbox{\bf Pr}}_{\Omega}\left[#1\right]}

\newtheorem{theorem}{Theorem}
\newtheorem{lemma}[theorem]{Lemma}
\newdef{definition}{Definition}
\newenvironment{proofof}[1]{{\bf Proof of #1:}}{$\qed$\par}

\begin{document}
\conferenceinfo{SIGMOD'10,}{June 6--11, 2010, Indianapolis, Indiana, USA.}
\CopyrightYear{2010}
\crdata{978-1-4503-0032-2/10/06}
\clubpenalty=10000
\widowpenalty = 10000

\title{Similarity Search and Locality Sensitive Hashing using Ternary Content Addressable Memories}
\numberofauthors{4}

\author{
\alignauthor
Rajendra Shinde \titlenote{Research supported by NSF award 0915040 and a gift from Lightspeed Venture Partners.} \\
\affaddr{Stanford University.}\\
\affaddr{Stanford, CA USA} \\
\email{rbs@stanford.edu}
\alignauthor Ashish Goel 
\titlenote{Research supported by NSF award 0915040 and a gift from Cisco Systems.} \\
\affaddr{Stanford University.}\\
\affaddr{Stanford, CA USA} \\
\email{ashishg@stanford.edu}
\alignauthor Pankaj Gupta 
\titlenote{Research supported by a gift from Cisco Systems. This work was done while the author was at Stanford University.}\\
\affaddr{Twitter Inc.} \\
\affaddr{San Francisco, CA USA} \\
\email{pankaj@cs.stanford.edu} 
\and
\alignauthor Debojyoti Dutta \\
\affaddr{Cisco Systems Inc.} \\ 
\affaddr{San Jose, CA, USA} \\
\email{dedutta@cisco.com}
}

%\subtitle{[Extended Abstract]
%\titlenote{A full version of this paper is available as
%\textit{Author's Guide to Preparing ACM SIG Proceedings Using
%\LaTeX$2_\epsilon$\ and BibTeX} at

\date{10 March 2009}
\maketitle
\begin{abstract}
Similarity search methods are widely used as kernels in various data mining and machine learning applications including those in computational biology, web search/clustering. Nearest neighbor search (NNS) algorithms are often used to retrieve similar entries, given a query. While there exist efficient techniques for exact query lookup using hashing, similarity search using exact nearest neighbors suffers from a "curse of dimensionality", i.e. for high dimensional spaces, best known solutions offer little improvement over brute force search and thus are unsuitable for large scale streaming applications. Fast solutions to the approximate NNS problem include Locality Sensitive Hashing (LSH) based techniques, which need storage polynomial in $n$ with exponent greater than $1$, and query time sublinear, but still polynomial in $n$, where $n$ is the size of the database. In this work we present a new technique of solving the approximate NNS problem in Euclidean space using a Ternary Content Addressable Memory (TCAM), which needs near linear space and has O(1) query time. In fact, this method also works around the best known lower bounds in the cell probe model for the query time using a data structure near linear in the size of the data base.

TCAMs are high performance associative memories widely used in networking applications such as address lookups and access control lists. A TCAM can query for a bit vector within a database of ternary vectors, where every bit position represents $0$, $1$ or $*$. The $*$ is a wild card representing either a $0$ or a $1$. We leverage TCAMs to design a variant of LSH, called Ternary Locality Sensitive Hashing (TLSH) wherein we hash database entries represented by vectors in the Euclidean space into $\{0,1,*\}$. By using the added functionality of a TLSH scheme with respect to the $*$ character, we solve an instance of the approximate nearest neighbor problem  with 1 TCAM access and storage nearly linear in the size of the database. We validate our claims with extensive simulations using both real world (Wikipedia) as well as synthetic (but illustrative) datasets. We observe that using a TCAM of width 288 bits, it is possible to solve the approximate NNS problem on a database of size 1 million points with high accuracy. Finally, we design an experiment with TCAMs within an enterprise ethernet switch (Cisco Catalyst 4500) to validate that TLSH can be used to perform 1.5 million queries per second per 1Gb/s port. We believe that this work can open new avenues in very high speed data mining. 

\end{abstract}

%\category{H.4}{Database Management System and Algorithm Designs for emerging hardware architectures: --}{Miscellaneous}
%A category including the fourth, optional field follows...
%\category{D.2.8}{Software Engineering}{Metrics}[complexity measures, performance measures]
%\terms{Delphi theory}

%\keywords{ACM proceedings, \LaTeX, text tagging}

\section{Introduction}

% Introduction to SS and usual implementations as nearest neighbor
Due to the explosion in the size of datasets and the increased availability of high speed data streams, it is has become necessary to speed up similarity search (SS), i.e. to look for objects within a database similar to a query object,  which is a critical component of most data mining and machine learning tasks. For example, consider searching for similar images within a corpus of billions of images and repeating this for a query set consisting of millions of images using as little power and computation time as possible. One could use the streaming model and stream the corpus over the latter set. In order to do this, one would typically deploy very fast computing devices or distribute it over several compute devices. In this paper, we show how this goal can be achieved with just an associative memory module,  ternary content addressable memory (TCAM) \cite{tcam:survey}, commonly used in networking for route lookups and access control list (ACL) filtering, to perform a specific variant of SS, i.e. determine the approximate nearest neighbor for the $L_2$ or Euclidean space. 

% SS--> nearest neighbor
Common tasks in mining and learning depend heavily on SS. For example, clustering algorithms are designed to maximize intra cluster similarity and minimize inter cluster similarity. In classification, the label of a new query object is determined based on its similarity to trained (labeled) data and their labels.  In several applications of SS such as in content based search, pattern recognition and computational biology, objects are represented by a large number of features in a high dimensional (metric) space, and SS is typically implemented using nearest neighbor search routines. Given a set consisting of $n$ points, the nearest neighbor search problem\cite{MP54} builds a data structure which, given a query point, reports the data point nearest to the query. For example, nearest neighbor methods and their variants have been used for classification purposes \cite{cover67}, stream classification \cite{JHan:stream:class} and clustering heuristics \cite{berkhin}. Applications of SS range from content search, lazy classifiers, to  genomics, proteomics, image search, and NLP \cite{Tzane,ravi:NLP:clustering,Dutta:Cheng,Buhler,google:video:lsh,Charikar:multiprobe,Brian:Kulis}.

%{\bf Solutions to SS and nearest neighbor:}
Existing solutions to the exact nearest neighbor problem offer little improvement over brute force linear search. The best known solutions to exact nearest neighbor include those which use space partitioning techniques like $k$-d trees \cite{Bentley:kdtree}, cover trees \cite{Kakade:covertrees}, navigating nets \cite{Krauthgamer:lee:navigating:nets}. However these techniques do not scale well with dimensions. In fact an experimental study\cite{expt:space:partitioning} indicates when number of dimensions is more than 10, space partitioning techniques are in fact slower than brute force linear scan.

% Enter LSH
A class of solutions that have shown to scale well are those that are based on locality sensitive hashing (LSH) \cite{indyk:survey} which solve the {\emph approximate} nearest neighbor problem. The $c$-Approximate Nearest Neighbor problem ($c$-ANNS) allows the solutions to return a point whose distance to the query is at most $c$ times the distance from the query to its nearest neighbor.
A family of hash functions is said to be locality-sensitive if it hashes nearby points to the same bin with {\it high} probability and hashes far-off points to the same bin with {\it low} probability. To solve the approximate nearest neighbor problem on a set of $n$ points in $d$ dimensional Euclidean space, the data points are hashed to a number of buckets using locality-sensitive hash functions in the pre-processing step. To perform a similarity search, the query is hashed using the same hash functions and the similarity search is performed on the data points retrieved from the corresponding buckets. %[Make crisper]
In the last few years LSH has been extensively used for SS in diverse applications including bioinformatics \cite{Buhler, Dutta:Cheng}, kernelized LSH in computer vision \cite{Brian:Kulis}, clustering \cite{indyk:clustering}, time series analysis \cite{indyk:timeseries}.
%Given parameters $(l,u,p_l,p_u)$, a $(l,u,p_l,p_u)$-LSH family of hash functions has the property that points separated by a distance at most $l$ are hashed to the same bin with probability at least $p_l$ and points separated by a distance at least $cl$ are hashed to the same bin with probability at most $p_u$. %[or statement 2:]
%{\bf Complexity of LSH and the need of hardware primitives}
%LSH based algorithms require space $O(n^{(1+\rho)})$ and have a query time $O(n^\rho)$ where $\rho = \log{p_l}/\log{p_u} $. %For Euclidean space Andoni and Indyk show a LSH family which achieves $\rho = O(1)/c^2$ \cite{AI06} and it is almost optimal considering the lower bound on LSH proved in by Motwani etal. \cite{MNP06}.
For the Euclidean space, the optimal LSH based algorithm which solves the $c$-ANNS problem has a space requirement of $O(n^{(1+1/c^2)})$ and a query time of $O(n^{(1/c^2)})$. For $c\approx 1$, this near quadratic space requirement of LSH and query time sub-linear(but still polynomial) in $n$, make it difficult to use LSH in streaming applications, especially at extremely high speeds, which are beyond the capability of a CPU. In such a scenario, we look for hardware primitives to accelerate c-ANNS. 

% enough LSH bashing, enter TCAM/TLSH
In this paper, we develop a variant of LSH, Ternary Locality Sensitive Hashing (TLSH), for solving nearest neighbor problem in large dimensions using TCAMs and we show that it is possible to formulate an almost 'ideal' solution to the c-ANNS problem with a space requirement near linear in the size of the data base and a query time of $O(1)$.
A TCAM is an associative memory where each data ``bit'' is capable of storing one of three states: 0,1,* which we denote as ternions. where * is a wildcard that matches both 0 and 1 \cite{tcam:survey}. Thus a TCAM can be considered to be a memory of $n$ vectors of $w$ ternions wide.The presence of wildcards in TCAM entries implies that more than one entry could match a search key. When this happens, the index of the highest matching entry (i.e. appearing at the lowest physical address) is typically returned. 
Access speeds for TCAMs are comparable to the fastest, most-expensive RAMs. For almost a decade, TCAMs have been used in switches and routers, primarily for the purposes of route lookup (longest prefix matching) \cite{netlogic,49} and packet classification \cite{netlogic,38}. In this paper, we present one application in which c-ANNS problem can be solved using a single TCAM lookup using a TCAM with width poly($\log{n}$) where n is the size of the database by using the TLSH family.
%Due to their large power consumption in comparison to SRAMs, TCAMs are still appropriate as hardware accelerators only for applications where they provide orders of magnitude speed improvement. 

%[Introduction of the TLSH family:]
For TLSH, we use ternary hash functions that hash any point in $\mathds{R}^d$ to the set {0,1,*}. Analogous to LSH, TLSH has property that nearby points are hashed to matching ternions with high probability. We obtain a TLSH family by partitioning $\mathds{R}^d$ using randomly oriented randomly translated parallel hyperplanes. Alternate regions between the hyperplanes are hashed to a *, while the remaining regions are hashed to 0,1,0,1.. alternately.
\begin{figure} \centering \label{intuitive:description}
\includegraphics[scale=0.55]{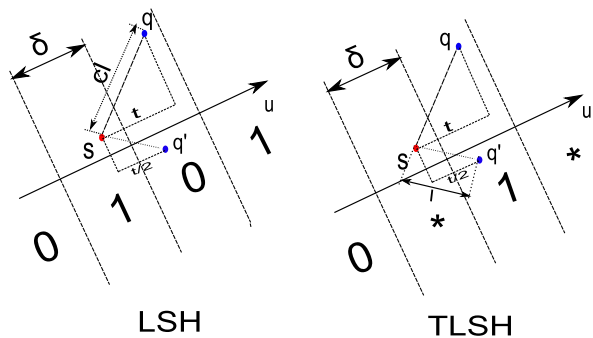}
\caption{A comparison of LSH and TLSH: We choose a random direction $\hat{u_i}$ and consider a family of hyperplanes orthonormal to it, adjacent hyperplanes being separated by $\delta$. The LSH family hashes regions between the hyperplanes to $0,1,0,1\ldots$, while the TLSH family hashes the regions between the hyperplanes to $0,*,1,*,0,*\ldots$}
%\hfill \includegraphics[scale = 0.25]{TLSH.pdf} \caption{TLSH}
\end{figure}

%[Intuition: explain why TLSH is an unbounded improvement]
In order to compare the TLSH family to the LSH family, we consider an example in which we choose a random direction $\hat{u_i}$ (say) and consider a family of hyperplanes orthonormal to it with adjacent hyperplanes separated by $\delta$ (say). The family of hyperplanes partitions $\mathds{R}^d$ as shown in figure \ref{intuitive:description}. We consider a LSH family which hashes the region between the hyperplanes to $0,1,0,1 \ldots$ and a TLSH family which hashes the region between the hyperplanes to $0,*,1,*,0,*\ldots$ as shown in figure \ref{intuitive:description}. 
Note that both LSH and TLSH project points in $\mathds{R}^d$ on to the random direction $\hat{u}$. Consider any two points ${\bf s,q} \in \mathds{R}^d$ (as shown in figure \ref{intuitive:description}), which are a distance $cl$ apart and whose projections on $\hat{u}$ are separated by $t$ (say) and another point ${\bf q'} \in \mathds{R}^d$ at a distance of $l$ from ${\bf s}$, whose projection on $\hat{u}$ is separated from that of ${\bf s}$ by $t/c$ (say). Ideally the notion of locality-sensitive hashing is aimed achieving the twin objectives of separating far-off points (hash them to opposite bits ) and hashing nearby points to matching bits with high probability. However, in this example we see that using a "binary" hash function from the LSH family, if the probability of separating ${\bf s}$ and ${\bf q}$ (hashing them to opposite bits) is $\psi(t)$ (say), then informally, the probability of separating ${\bf s}$ and ${\bf q'}$ is at least $\psi(t)/c$. Thus the "binary" hash function does not achieve both objectives simultaneously. On the other hand, if we set the distance between the hyperplanes of the TLSH family to value more than $t/c$, any function from the TLSH family will not separate ${\bf s}$ and ${\bf q'}$. This is because, any choice of translated hyperplanes will ensure that one of the following always happens: \begin{enumerate}
\item either ${\bf s}$ and ${\bf q'}$ are both hashed to 0 or both to 1.
\item One of ${\bf s}$ and ${\bf q'}$ is hashed to a *.
\end{enumerate} 
Thus the ternary hash representations of ${\bf s}$ and ${\bf q'}$ always match. In this manner, the regions hashed to a * "fuzz" the boundaries between regions hashed to $0$'s and $1$'s such that ternary hashed representation of nearby points match with high probability.

%[Intuition regarding the use of a TCAM: and a final precise statement spelling out the result]
We leverage the ability of a TCAM to represent the wildcard character (*) in order to implement TLSH and store the ternary hash signatures generated by it. Also, using the property of the TCAM of returning the highest matching entry, it is possible to configure this TCAM so that it solves a sequence of the $(1,c)$-Near Neighbor problem, which is a decision version of the $c$-ANNS problem, thus leading to a solution of the $c$-ANNS problem itself (details in section \ref{canns}). Hence, using the TLSH family of hash functions along with a TCAM of width poly$(\log{n})$ where $n$ is the size of the database, the c-ANNS problem can be solved in a single TCAM lookup [details in Sec \ref{canns}]. We believe this observation is very promising with regard to solving similarity search problems in streaming environments. Also we note that Tao \etal \cite{tao} describe a novel method to solve the $c$-ANNS problem without solving a sequence of near neighbor problems. Their method involves computation of longest common prefixes of binary strings. It would be interesting to know if their methods can be adapted for use with TCAMs in order to avoid solving a sequence of the $(1,c)$-Near Neighbor problems. 

%[A small note regarding beating lower bounds: the power of parallel lookups]
We note that this method beats the lower bounds for c-ANNs in the cell probe model according to which, any data structure nearly linear in n needs $\Omega(\log{n}/\log{\log{n}})$ probes in the data base in order to answer approximate nearest neighbors accurately \cite{PTW08}. This is because the TCAM implicitly implements highly parallel operations which do not conform to the cell probe model of computation. 
%[simulations:]

We also present simulations which explore the space of design parameters and establish the trade-off involved between the size of the TCAM used and the performance of our algorithm. We use a combination of real world and artificially generated data sets each containing one million points in a $64$ dimensional Euclidean space. The first data set contains randomly generated points (from a suitably chosen localized region), the second one contains simHash signatures of web pages belonging to the English Wikipedia (from a snapshot of the English Wikipedia in 2005), and the third one is again artificially generated in order to maximize the number of false positives and false negatives, by having many data points on the threshold of being similar or dissimilar to a query. From our simulations we observe that a TCAM of width 288 bits solves the decision version of the $2$-Approximate Nearest Neighbor problem accurately for the aforementioned databases. 

%[experiments:]
In order to validate our simulations, we design a novel experiment using TCAMs within a CISCO Catalyst 4500 Ethernet switch
and high speed traffic generators. We demonstrate how one can process approximately 1.5M approximate nearest neighbor queries per second for each port. 
%[Comparison to GPUs, FPGAs]
Thus, it is technically feasible to build devices with TCAMs that could serve as high speed similarity engines, in a vein similar to using GPUs to accelerate certain classes of application. Note that in our case, a TCAM is much more suitable due to the combined implicit memory access (lookup) and wildcard search done in parallel.  

\subsection{Organization}
\sloppy In section \ref{prelim}, we define the the $c$-Approximate Nearest Neighbor problem, the $(l,c)$-Near Neighbor problem,  and a $(l,u,p_l,p_u)$-TLSH family. In section \ref{da} we describe the construction and analysis of a \begin{math}(1,c,p_1(\delta),p_2(\delta/c))\end{math}-TLSH family for any $\delta \geq 0$. Section \ref{appnns} describes the use of a $(1,c,p_1(\delta),p_2(\delta/c))$-TLSH family in solving the $(1,c)$-Near Neighbor problem, $(1,c)$-Similarity Search problem and $c$-Approximate Nearest Neighbor problem. Section \ref{simu:expt} describes simulations using a combination of real life and synthetic data sets containing a million points in $64$ dimensional space, which explore the trade-off between the width (size) of the TCAM and the performance of the method, along with experiments which validate our results. Section \ref{relwork} summarizes the related work. We summarize the findings of this paper in section \ref{conclusion}.

%\section{Introduction}
\section{Preliminaries}\label{prelim}

First we define the $c$-Approximate Nearest Neighbor Search problem .
\begin{definition} $c$-Approximate Nearest Neighbor Search or the $c$-ANNS problem: \\
Given a set $S$ of $n$ points in $\mathds{R}^d$, construct a data structure which, given a query $q \in \mathds{R}^d$ returns a point $s \in S$ whose distance from $q$ is at most $c$ times the distance between $q$ and the nearest neighbor of $q$ in $S$.  
\end{definition}
Next, we define the TCAM match operation "$=_T$'' which declares that two sides match if both are equal or
one of them is a $*$.
\begin{definition} If $A,B \in \{0,1,*\}$, then $A{=}_TB$ if and only if $A = B \text{ or } A = * \text{ or } B = *$. The complementary relation is referred to as $\neq_T$.
\end{definition}

%Let $T_l = \{x \in S: \norm{x-q} \leq l\}$ and $T_u = \{ x \in S : \norm{x-q} \geq cl\}$. Then the $(l,c)$- SS problem is to report $T \subseteq S$, such that $T_l \subseteq T$ and $T_u \cap T = \phi$. We solve the $\eta$- randomized version of this problem, i.e. output $T$ such that $\forall x \in T_l$, $x \in T$ w.p. atleast $1 -\eta$ and $T \cap T_u = \phi$ w.p. atleast $1 - \eta$. The former is equivalent to saying that the false negative probability $p^{fn} \leq \eta$. The latter is equivalent to saying that the false positive probability is atmost $\epsilon$. Now, let $p^{fp}$ denotes the probability that for some $y \in T_u$, we have $y \in T$. (Note that $p^{fp}$ denotes the false positive probability with respect to a specific point in $S$.) If $p^{fp} \leq \eta/n$ and using $|T_u| \leq n$, the expected number of false positives is atmost $\eta$. Then using Markov's inequality, we can argue that the probability of having any false positive is at most $\eta$ which shows that $p^{fp} \leq \eta/n$ suffices to satisfy the latter.  Thus a data structure which has $p^{fn} \leq \epsilon$ and $p^{fp} \leq \frac{\epsilon}{n}$ solves the $(l,c)$- SS problem. Note that we can scale down the distances between the points by $l$ in which case we need to solve the above problem for $l = 1$. Thus we will consider $l = 1$ for the rest of the paper and refer to this problem as $(1,c)$- Similarity Search problem. 
\begin{definition} $(l,c)$-Near Neighbor problem or the $(l,c)$-NN problem: \\
Given a set $S$ of $n$ points in $\mathds{R}^d$, construct a data structure which, given a query point $q \in \mathds{R}^d$,  if there exists a point $s_l \in S$ such that $\norm{s_l-q} \leq l$, then reports ``Yes'' and a point $s$ such that $\norm{s-q} \leq cl$ and if there exists no point $s_u$ such that $\norm{s_u-q} \leq cl$ then reports ``No''.
\end{definition}
Note that we can scale down all the coordinates of points by $l$ in which case the above problem needs to be solved only for $l = 1$. Accordingly we discuss the solution of $(1,c)$-NN problem in section \ref{appnns}. Also, note that the $(l,c)$-NN problem is the decision version of the $c$-ANNS problem. The $c$-ANNS problem can be reduced to $O(\log{\frac{n}{c-1}})$ instances of $(1,c)$-NN problems \cite{HP01}. Next analogous to \cite{IM98}, we define a ternary locality sensitive hashing family. % and refer to this problem as $(1,c)$-Near Neighbor problem. 
\begin{definition} \sloppy Ternary locality sensitive hashing family (TLSH): \\
A distribution $\Omega$ on a family $\mathcal{G}$ of ternary hash functions $($i.e. functions which map
$\mathds{R}^d \to \{0,1,*\})$ is said to be an $(l,u,p_l,p_u)$-TLSH if $\forall \; x,y \in \mathds{R}^d$
\begin{equation*}\begin{array}{l}
\text{if } \norm{x-y} \leq l \text{ then } \probomega{g(x)=_Tg(y)} \geq p_l, \\ 
\text{if } \norm{x-y} \geq u \text{ then } \probomega{g(x)=_Tg(y)} \leq p_u. \\
\end{array}\end{equation*}
where $g$ is drawn from the distribution $\Omega$.
\end{definition}

%\fussy Similar to the definition of $(l,c)$-Near Neighbor problem, we define the $(l,c)$-Similarity Search problem where the data structure is supposed to return all data points which match a given query point. 

%\begin{definition} $(l,c)$-Similarity Search(SS) Problem: \\
%Given a set $S$ of $n$ points in $\mathds{R}^d$  and a query point $q \in \mathds{R}^d$,
%report a set $T \subseteq S$ such that $\{x \in S: \norm{x-q} \leq l\} \subseteq T$ and $\{ x \in S :
%\norm{x-q} > cl\} \cap T = \phi$.
%\end{definition}
\section{\sloppy Design and Analysis of a TLSH family}\label{da}
In this section we describe the construction and analysis of a TLSH family. We will show its application in solving the $(1,c)$-NN problem in section \ref{appnns}. Let $\delta > 0$ be any constant. Next, we describe the construction of a family of ternary hash functions $\mathcal{G}_{\delta} = \{g_{\delta} : \mathds{R}^d \to \{0,1,*\}\}$. 

Each hash function in the family $\mathcal{G}_{\delta}$ is indexed by a random choice of ${\bf a}$ and $b$ where ${\bf a} \in \mathds{R}^d$, individual components $a_i$ of ${\bf a}$, $i = 1\ldots d$ are chosen independently from the normal distribution $\mathcal{N}(0,1)$ where $\mathcal{N}(\mu,\sigma)$ denotes a normal distribution of mean $\mu$ and variance $\sigma^2$, and $b$ is a real number chosen uniformly from $(0,2\delta)$. We represent each hash function  in the family $\mathcal{G}_{\delta}$ as $g_{\delta,{\bf a},b}: \mathds{R}^d \to \{0,1,*\}$ and $g_{\delta, {\bf a}, b}$ maps a $d$ dimensional vector onto the set $\{0,1,*\}$. For sake of convenience, we drop the subscript $\delta$ from $g$ and refer to it as $g_{{\bf a},b}$ which is defined as follows. Given ${\bf a},b$, for any $x \in \mathds{R}^d$, let $j = \lfloor \frac{{\bf x}.{\bf a}+b}{\delta} \rfloor\text{mod}(4)$ where $\text{mod}$ denotes the modulus function.
\begin{equation*}
\begin{array}{ll}
\text{if } j = 0& g_{{\bf a},b}({\bf x}) = 0  \\
\text{if } j = 2& g_{{\bf a},b}({\bf x}) = 1  \\
\text{if } j = 1 \text{ or 3}& g_{{\bf a},b}({\bf x}) = *  \\
\end{array}
\end{equation*}

Having given a formal definition, we give an intuitive description this family of hash functions. 
Consider a partition of the space $\mathds{R}^d$ due to the family of hyperplanes orthonormal to $\bf a$, adjacent planes separated by $\delta$ and randomly shifted from the origin by $-b$. Then the function $g_{{\bf a},b}({\bf x}):\mathds{R}^d \to \{0,1,*\}$ hashes alternate regions to $*$ and the remaining regions are hashed to $0,1$ alternately. We show in this section that $\mathcal{G}_{\delta}$ is a TLSH family with parameters that are exponentially better than LSH. We show applications of this scheme in section \ref{appnns}. 

Next, we state the following theorem which is the main technical contribution of this paper.
\begin{theorem} \label{main}For all $\delta \geq 2c$ the family of ternary hash functions $\mathcal{G_{\delta}}$ is a $(1,c,p_1(\delta),p_2(\frac{\delta}{c}))$-TLSH family where $ p_1(z) = 1 -\frac{1}{z^3\sqrt{2\pi}}e^{-\frac{z^2}{2}}$, and $p_2(z) = 1-\frac{1}{5z^3\sqrt{2\pi}}e^{-\frac{z^2}{2}}$. \end{theorem}

Before proving this theorem, we comment on the improvement that an $(1,c,p_1(\delta),p_2(\delta/c))$-TLSH family offers over a $(1,c,p_l,p_u)$-LSH family. One way to compare the two hashing schemes is to compare the values of $\rho:= \log{p_l}/\log{p_u}$. We note that when $\delta$ is large, both $p_l$ and $p_u$ are close to $1$. In fact, for applications in section \ref{appnns} we set $\delta = O(\sqrt{\log{\log{n}}})$. Hence in this regime we can use the the approximations $\log{1/p_l} \approx 1 - p_l$, $\log{1/p_u} \approx 1 - p_u$. We get $\log{p_l}/\log{p_u} \approx (1/c^3)e^{-\delta^2(c^2-1)/(2c^2)}$. If we set $\delta = O(\sqrt{\log{\log{n}}})$  then $\rho$ decreases to $0$ unbounded as a function of $n$. On the other hand, Motwani \etal have proved that $\log{p_l}/\log{p_u} \geq 0.5/c^2$ for any LSH family \cite{MNP06}. Hence we get an unbounded improvement in the parameter $\rho$ of a locality sensitive hashing family by using TLSH in the range of parameter $\delta$ which is of interest.

In order to prove Theorem 3.1 we first introduce some notation and prove some subsidiary lemmas. The applications in section \ref{appnns} use the statement of the theorem but are independent of the proof.

Consider two points ${\bf s}$, ${\bf q}$ in $\mathds{R}^d$. Let ${\bf x} = {\bf s-q}$, $x = \norm{{\bf x}}$. Let $\Psi({\bf x})$  denote the ``collision probability'' of $s$ and $q$, i.e. $\Psi({\bf x}) = \probomega{g({\bf s})=_Tg({\bf q})\text{ }|\text{ }{\bf s-q= x }}$ and $\psi(t)$ denote the collision probability conditioned on the fact that $|{\bf a \cdot x }| = t$, i.e. $\psi(t) = \prob{g({\bf s}) =_T g({\bf q})\text{ }|\text{ }|{\bf a \cdot x }| = t}$.  We have  $\Psi({\bf x}) = \int_0^{\infty} \psi(t)\pi_x(t)dt$ where $\pi_x(t)$ is the density of the random variable $|{\bf a \cdot x }|$. Let $\bar{\Psi}({\bf x}) = 1 - \Psi({\bf x}) $ and $\bar{\psi}(t) = 1 - \psi(t)$. Let $\bar{F}$ denote the complementary cumulative distribution function of $\mathcal{N}(0,1)$.  Let $\phi(y)  = e^{-(y^2/2)} /(y\sqrt{2\pi}) - \bar{F}(y)$. The following lemma proves lower and upper bounds on the collision probability $\Psi$. 

\begin{lemma}\label{sim} For all $\delta > 0$, 
$\phi(\frac{\delta}{x})-\phi(\frac{2\delta}{x}) \leq \bar{\Psi}({\bf x}) \leq \phi(\frac{\delta}{x})$. 
\end{lemma}
\begin{proof}
First we recall the definition of stability of random variables.
A distribution $\mathcal{D}$ over $\mathds{R}$ is called $p$-stable if there exists $p \geq 0$ such that for any $n$ real numbers $v_1,v_2\ldots v_n$ and i.i.d. random variables $X_1,X_2,\ldots X_n$ with distribution $\mathcal{D}$, the random variable $\sum_i v_iX_i$ has the same distribution as the variable ${(\sum_i|v_i|^p)}^{1/p}X$, where $X$ is a random varible with distribution $\mathcal{D}$. Using the well known fact that the Normal distribution $\mathcal{N}(0,1)$ is $2$-stable, we conclude that the random variable $|{\bf a \cdot x}|$ is distributed as $x \cdot |\mathcal{N}(0,1)|$ which implies $\pi_x(t) = \frac{\sqrt{2}}{x\sqrt{\pi}}e^{-\frac{t^2}{2x^2}}$. 

If two points $\bf s$, $\bf q$, are such that $|{\bf a \cdot x}| \leq \delta$ then they are hashed to matching TCAM values, i.e. $g_{{\bf a},b}( {\bf s}) =_T g_{{\bf a},b}({\bf q})$, since the adjacent hyperplanes are at a distance of $\delta$ from each other and alternate regions are hashed to $*$ and $0,1,0,1\ldots$. Hence $\bar{\psi}(t) = 0$ if $t \in (0,\delta)$. In fact, if $t \in (0,2\delta)$, then $\bar{\psi}(t) = \frac{t-\delta}{2\delta}1_{t > \delta}$,  where $1_{t > \delta}$ is the indicator function ($1$ if $t > \delta$, $0$ otherwise). Symmetry about $2\delta$ implies that if $t \in (2\delta,4\delta)$, then $\bar{\psi}(t) = \left(\frac{3\delta- t}{\delta}\right)1_{t \leq 3\delta}$. Also note that the function $\psi(t)$ is periodic with period $4\delta$. So $\bar{\psi}(t+4k\delta) = \bar{\psi}(t)$, for all positive integers $k$. 

Now $\bar{\Psi}({\bf x}) = \int_{\delta}^{\infty}\bar{\psi}(t)\pi_x(t)dt >  \int_{\delta}^{2\delta}\bar{\psi}(t)\pi_x(t)dt$. Using $\bar{\psi}(t) = \frac{t-\delta}{2\delta}$ when $t \in (\delta,2\delta)$, we get \begin{equation*}
\bar{\Psi}(x) > \sqrt{\frac{2}{\pi}}\int_{\delta}^{2\delta}\frac{t-\delta}{2\delta}\frac{e^{-\frac{t^2}{2x^2}}}{x}dt = \phi(\frac{\delta}{x}) - \phi(\frac{2\delta}{x}).\end{equation*}

To prove the latter inequality, we use the fact that $\bar{\psi}(t) \leq \frac{t-\delta}{2\delta}$,  $\forall t > \delta$. Hence we have that 
\begin{equation*} \begin{array}{ll}
 \bar{\Psi}(x) &= \int_{\delta}^{\infty}\bar{\psi}(t)\pi_x(t)dt \\
\bar{\Psi}(x) &\leq  \frac{\sqrt{2}}{\sqrt{\pi}}\int_\delta^\infty\left(\frac{t-\delta}{2\delta}\right) \frac{1}{x}e^{-\frac{t^2}{2x^2}}dt \\
&= \frac{1}{\sqrt{2\pi}}\frac{e^{-\frac{\delta^2}{2x^2}}}{\frac{\delta}{x}} - \bar{F}\left(\frac{\delta}{x}\right) \\
&= \phi \left(\frac{\delta}{x}\right).
\end{array}
\end{equation*}
\end{proof}

The lemma \ref{lem-bound-phi} specifies appropriate bounds for the function $\phi$.

\begin{lemma}\label{lem-bound-phi}The function $\phi$ is bounded above and below as follows:
\begin{equation}
\frac{1}{4\sqrt{2\pi}}\frac{e^{-\frac{y^2}{2}}}{y^3} \leq \phi(y) \leq \frac{1}{\sqrt{2\pi}}\frac{e^{-\frac{y^2}{2}}}{y^3}, \label{bound-phi} \\
\end{equation}
where the first inequality holds if $y \geq 2$. Hence for all $y \geq 2$,  \begin{equation}
\phi(y) - \phi(2y) \geq \frac{1}{5\sqrt{2\pi}}\frac{e^{-\frac{y^2}{2}}}{y^3}. \label{lbub} \end{equation}
\end{lemma}

\begin{proof}
The expansion of the error function using integration by parts \cite{AS72} proves \eqref{bound-phi}.
Using \eqref{bound-phi} we get $\phi(y)-\phi(2y) \geq (\frac{1}{4}-\frac{1}{8}e^{-\frac{3y^2}{2}})\frac{1}{y^3\sqrt{2\pi}}e^{-\frac{y^2}{2}}$. Using $y \geq 2$ proves the lemma.
\end{proof}

Now we return to the proof of the main theorem. \\

\begin{proofof}{Theorem 3.1}

Using lemma \ref{sim} and \eqref{bound-phi}, $\forall x \leq 1$, we have \begin{equation} \label{upperb-psi-loose}
\bar{\Psi}(x) \leq \phi\left(\frac{\delta}{x}\right)\leq \frac{1}{\sqrt{2\pi}}\frac{x^3}{\delta^3}e^{-\frac{\delta^2}{2x^2}} \leq \frac{1}{\sqrt{2\pi}}\frac{e^{-\frac{\delta^2}{2}}}{\delta^3}. \end{equation}
Also using lemma \ref{sim} and \eqref{lbub}, $\forall x \geq c$ and $\delta \geq 2c$, we have \begin{equation}\begin{array}{ll} \label{lowerb-psi-loose} \bar{\Psi}(x) &\geq \phi\left(\frac{\delta}{x}\right)-\phi\left(\frac{2\delta}{x}\right) \\  &\geq \phi\left(\frac{\delta}{c}\right)-\phi\left(\frac{2\delta}{c}\right) \\
&\geq  \frac{1}{5\sqrt{2\pi}}\frac{c^3}{\delta^3}e^{-\frac{\delta^2}{2c^2}}.\end{array}\end{equation} This proves the theorem \ref{main}.

\end{proofof}

Note that using standard bounds on the complimentary cumulative distribution function of the standard normal random variable $\mathcal{N}(0,1)$  \cite{AS72}, the bounds on $\phi$ can be improved as follows: $\forall y$, we have \begin{equation} \begin{array}{l} \label{bound-phi_tight}\left(\frac{\frac{4}{\pi}}{y^2+y\sqrt{y^2+\frac{8}{\pi}}+\frac{4}{\pi}}\right)\frac{e^{-\frac{y^2}{2}}}{y\sqrt{2\pi}} \leq \phi(y) 
 \leq \left(\frac{2}{y^2+y\sqrt{y^2+4}+2}\right)\frac{e^{-\frac{y^2}{2}}}{y\sqrt{2\pi}}.\end{array}\end{equation} It can be verified using standard plotting packages like Maple or Matlab that for small values of $n$ and $1/\epsilon$ these bounds are in fact tighter than the bounds presented in \eqref{bound-phi}. However it is not clear how these stronger bounds can be used to obtain an improvement in Theorem 3.1. Analysis using these bounds is complicated and moreover, asymptotically these bounds have the same behaviour as the bounds in \eqref{bound-phi}. Hence we present the analysis using simpler bounds as presented in the lemma's above but we recommend the use of tighter bounds for parameter tuning and experiments as illustrated in section \ref{simu:expt}. 

\section{Approximate Similarity Search}\label{appnns}
In this section we demonstrate the use of $\mathcal{G_{\delta}}$ to solve the $(1,c)$-NN problem and the $c$-ANNS problem on a set $S$ consisting of $n$ points in $\mathds{R}^d$ using a TCAM of width $w$ for some appropriate choice of parameters $\delta$ and $w$. We note that the results of this section can be extended to solve the $(1,c)$-SS problem by requiring the TCAM to output all matching data points to a query point. 

\subsection{The $(1,c)$-NN Problem}\label{1cnn}
In this section we formulate an algorithm to solve the $(1,c)$-NN problem. The choice of parameters $\delta$ and $w$ is specified later. \\

\noindent {\bf Algorithm A}
\begin{itemize}

\sloppy \item {\bf Pre-processing (TCAM Setup):} Choose $w$ independent hash functions $g_1,g_2, \ldots g_w \in \mathcal{G}_{\delta}$ where $\mathcal{G}_{\delta}$ is a $(1,c,p_1(\delta),p_2(\delta/c))$-TLSH family as defined in section \ref{da}. For every ${\bf s}_i \in S$, find its TCAM representation $T(s_i):=(g_1({\bf s}_i),g_2({\bf s}_i), \ldots g_w({\bf s}_i))$. 
\item {\bf Query lookup:} Given a query ${\bf q}$ find its TCAM representation $T({\bf q})$ (using the same hash functions). Perform a TCAM lookup of $T({\bf q})$. If the TCAM returns a point ${\bf s_T}$ such that $\norm{\bf q-s_T}\leq c$, return ``YES'' and ${\bf s_T}$, otherwise return ``NO''. 
\end{itemize}
%Note that this algorithm can be generalized to solve the $(l,c)$-NN problem using a $(l,u,p_l,p_u)$-TLSH family with $u = cl$.
\fussy Intuitively choosing a large $w$ (i.e. a large no. of hash functions) reduces the possibility of having false positives in the output but at the same time increases the chances of a false negative occurring because any one (or more than one) of the $w$ TCAM ternions can produce a false negative. Choosing a large value of $\delta$ reduces the false negative probability but increases the likelihood of having false positives. We show in the following theorem that it is possible to tune these parameters simultaneously to ensure that the false negative probability is small and the expected number of false positives is also small.

\begin{theorem} Consider a set $S$ consisting of $n$ points in $\mathds{R}^d$. \label{thm4.1}
\begin{enumerate}
\sloppy \item {\bf {\it One} TCAM lookup:} The $(1,c)$-NN problem can be solved by using a TCAM of width $w$ where  
$w = O({\left( \frac{1}{\epsilon} \log{\frac{n}{\epsilon}}\right)}^{\frac{c^2}{c^2-1}}{\left(\log{\left(\frac{1}{\epsilon}\log{\frac{n}{\epsilon}}\right)}\right)}^{\frac{3}{2}}(c^2-1)^{-\frac{3}{2}})$with error probability at most $\epsilon$ using exactly $1$ TCAM lookup and $1$ distance computation in $\mathds{R}^d$. 
\item {\bf $O(\log{(1/\epsilon)})$ TCAM lookups:} The $(1,c)$-NN problem can be solved by a TCAM of width $w = O({(\log{n})}^{\frac{c^2}{c^2-1}}{(\log{\log{n}})}^{\frac{3}{2}}(c^2-1)^{-\frac{3}{2}})$ with error probability at most $\epsilon$ using  $O\left(\log{\frac{1}{\epsilon}}\right)$ TCAM lookups and $O\left(\log{\frac{1}{\epsilon}}\right)$ distance computations in $\mathds{R}^d$.  
\item {\bf Word size $O(\log{n})$:} If $c^2 \geq \log{\left( \frac{k}{\epsilon} \log{\frac{n}{\epsilon}}\right)}$ where $k \geq 1/(1-p_2(2))$, a constant, the $(1,c)$-NN problem can be solved with error probability at most $\epsilon$ using a TCAM of width $k\log{(n/\epsilon)}$.
\end{enumerate}
\end{theorem}

\fussy Before proving Theorem 4, we discuss the improvements it provides over existing methods to solve the $(1,c)$-NN problem.   
\begin{enumerate}
\item{{\bf Constant separation} $c$:}
\sloppy Existing approaches to solve the $(1,c)$-NN problem can be broadly classified into three categories depending on their space requirements as a function of $n$: polynomial, sub quadratic, and near linear. Using the dimensionality reduction approach proposed by Ailon and Chazelle \cite{AC06} and ignoring the dependence on $\epsilon$, it is possible to solve the $(1,c)$-NN problem with a query time of $O(d\log{d}+(c-1)^{-3}\log^2{n})$ using a data structure of size $d^2n^{O(1/{(c-1)}^2)}$ i.e. polynomial in $n$. The space requirement of $n^{O(1/{(c-1)}^2)}$ is optimal in the sense that any data structure which solves $(1,c)$-NN problem with a constant number of probes must use $n^{\Omega(1/{(c-1)}^2)}$ space \cite{AI06}. However, the extremely large space requirement when $c$ is close to $1$ seems to render this approach impractical. An alternative approach based on the optimal LSH family \cite{AI06} proposed by Andoni and Indyk can be used to solve the $(1,c)$-NN Problem using a data structure with sub quadratic space requirement and a constant probability of success. Their approach has a query time of $O(dn^{1/c^2})$ and space requirement of $O(dn^{1+1/c^2}\log{n})$ where the dependence on $\epsilon$ has been ignored. To the best of our knowledge, their algorithm minimizes the query time when the size of the data structure is limited to be sub quadratic in $n$. The optimal LSH family \cite{AI06} can also be used to formulate an algorithm which solves the $(1,c)$-NN problem with a data structure which is near linear in size and has a query time of $dn^{O(1/c^2)}$, using the algorithm proposed by Panigrahy \cite{P06}. These upper bounds reveal the trade off involved between the space requirement and the query time while solving the $(1,c)$-NN problem using LSH. In contrast with these results using [Theorem \ref{thm4.1},1], we can formulate a TCAM based data structure which has $O({\left( \frac{1}{\epsilon}\log{\frac{n}{\epsilon}}\right)}^{\frac{c^2}{c^2-1}}{\left(\log{\left(\frac{1}{\epsilon}\log{\frac{n}{\epsilon}}\right)}\right)}^{\frac{3}{2}}(c^2-1)^{-\frac{3}{2}})$ word size and solves the $(1,c)$-NN problem in just {\it one} TCAM lookup and {\it one} distance computation in $\mathds{R}^d$. Ignoring the dependence on $\epsilon$, we conclude that a TCAM based data structure requires word size $O({\log{n}}^{c^2/(c^2-1)}{(\log{\log{n}})}^{3/2}(c^2-1)^{-3/2})$ to solve the $(1,c)$-NN problem with query time $O(1)$. The width of the TCAM varies with $\epsilon$ as $\epsilon^{-c^2/(c^2-1)}$ which leads to large values of the width when $\epsilon$ is small. One work around is to use a TCAM of width $O({(\log{n})}^{c^2/(c^2-1)}{(\log{\log{n}})}^{3/2}(c^2-1)^{-3/2})$ and repeat the algorithm  $O(\log{(1/\epsilon)})$ times [Theorem \ref{thm4.1},2]. For instance, $n = 10^6$ and $c=2$ requires a TCAM of width $3.3$K bits and $1$ lookup per query to succeed with probability $90\%$ using the tight bounds in \eqref{bound-phi_tight}. But allowing $4$ lookups per query, the width of the TCAM required can be brought down to $1.7$K bits. We explore the trade-off between the width of the TCAM and accuracy of algorithm {\bf A} while using data sets consisting of a $n=10^6$ points in a practical setting in section \ref{simu:expt}. 

In fact, Panigrahy \etal \cite{PTW08} showed that any data structure in the cell probe model \cite{Yao:cellprobe} which uses a single probe to solve the $(1,c)$-NN  problem with constant probability has a space requirement of $n^{1+\Omega(1/c^2)}$. Hence a data structure which uses near linear space needs to be probed $\Omega\left(\log{n} / \log{\log{n}}\right)$ times. Clearly, the TCAM based scheme which uses space $\tilde{O}(n))$ and query time $O(1)$ beats this lower bound by implementing parallel operations which do not conform with the cell probe model of computation.

\item{{\bf Word size} $O(\log{n})$:}
\fussy Consider solving the $(1,c)$-NN  problem using a RAM of word size $w = O(\log{n})$ which uses $w$ independent hash functions from the optimal $(1,c,p_l,p_u)$-LSH family \cite{AI06}. To solve the $(1,c)$-NN  problem with error probability at most $\epsilon$, we need the probability of a false negative to be at most $\epsilon$, i.e.  $1-p_l^w \leq \epsilon$ and the probability of a false positive to be at most $\epsilon$ i.e.  $p_u^w \leq \epsilon/n$ (since there are at most $n$ points with respect to which a false positive can occur). This implies that $ \frac{\log{p_u}}{\log{p_l}}  \geq (\frac{1}{\epsilon} -1) \log{\frac{n}{\epsilon}}$. Hence $ \frac{\log{p_u}}{\log{p_l}} \geq  \Omega\left(\frac{1}{\epsilon}\right)\log{(n/\epsilon)}$. Using the fact that $\frac{\log{p_u}}{\log{p_l}} \geq \frac{0.46}{c^2} $ \cite{MNP06} for any $(l,u,p_l,p_u)$-LSH family, we get $c^2 = \Omega\left(\frac{1}{\epsilon}\log{\frac{n}{\epsilon}}\right)$. Hence ``granularity'' achieved by LSH (ignoring $\epsilon$) in this case $\Omega(\sqrt{\log{n}})$.
On the other hand using [Theorem \ref{thm4.1},3] using a word size of $O(\log{n})$, algorithm {\bf A} can solve the $(1,c)$-NN problem with error probability at most $\epsilon$ if 
$c = \Omega\left(\sqrt{\log{\left( \frac{1}{\epsilon} \log{\frac{n}{\epsilon}}\right)}}\right)$. Thus, ignoring $\epsilon$, the granularity achieved by TCAM based scheme is $\Omega(\sqrt{\log{\log{n}}})$. Hence we see that use of TLSH family brings about an exponential improvement in the ``granularity'' of a $(1,c)$-NN problem.
\end{enumerate}
Again, we note that these huge improvements are brought about by the use of a TCAM which has a lot of inherent parallelism and hence the lower bounds mentioned before do not apply. Next we proceed to prove Theorem 4. \\ 
\\
\begin{proofof}{Theorem 4}\label{analysis:thm4.1}
First we make the following claims regarding the choice of parameters $\delta,w$ which prove the theorem. 
\begin{enumerate}
\item { \bf {\it One} TCAM lookup:} \sloppy If we choose $\delta = {\left(\frac{2c^2}{c^2-1}\log{\left(\frac{10}{c^3\epsilon}\log{\left(\frac{2n}{\epsilon}\right)}\right)}\right)}^{1/2}$ and $w = \frac{1}{1-p_2(\delta/c)}\log{\left(\frac{2n}{\epsilon}\right)}$, then algorithm {\bf A} solves the $(1,c)$-NN problem with error probability at most $\epsilon$.  
\item {\bf $O(\log{(1/\epsilon)})$ TCAM lookups:} Choosing $\delta$ and $w$ as in [Theorem 4,1] with error probability at most $1/2$ and repeating algorithm {\bf A} $\log{(1/\epsilon)}$ times solves the $(1,c)$-NN problem with error probability at most $\epsilon$. This can in fact be implemented using a single TCAM by using the first $O(\log{(\log{(1/\epsilon)})})$ bits of the TCAM to code the version number of the $O(\log{(1/\epsilon)})$ different data-structures to be used to solve the $(1,c)$-NN problem.  
\item {\bf Word size $O(\log{n})$:} Choosing $\delta = \alpha c$ and $w = k\log{\left(\frac{2n}{\epsilon}\right)}$ where $\alpha$ is such $k = \frac{1}{1-p_2(\alpha)}$ implies algorithm {\bf A} solves the $(1,c)$-NN problem with error probability at most $\epsilon$ when $c^2 \geq \log{\left(\frac{k}{\epsilon}\log{\frac{n}{\epsilon}}\right)}$.  
\end{enumerate}

Next we prove these claims sequentially. Let $S({\bf q},c)$ denote the set of points $\{{\bf s} \in S \text{, }\norm{{\bf s}-{\bf q}} > c\}$. We prove the theorem by analyzing the false positive and false negative cases.  
For any query point $q \in \mathds{R}^d$, note that algorithm {\bf A} will solve the $(1,c)$-NN problem correctly if the following two properties hold:
\begin{description}
\item[{\bf P1:}]{(No false negative matches)} If there exists a ${\bf s}_L$ such that $\norm{{\bf s_L}-{\bf q}}\leq 1$ then TCAM representations of ${\bf s}_L$ and ${\bf q}$ match. i.e. $T({\bf s_L}) =_{\text{T}} T({\bf q})$. 
\item[{\bf P2:}]{(No false positive matches)} For any $s_U \in S({\bf q},c)$, TCAM representations of ${\bf s_U}$ and ${\bf q}$ do not match, i.e. $T({\bf s_U}) \neq_{\text{T}} T({\bf q})$. 
%The total number of points of $S({\bf q},c)$ with TCAM representations which match that of ${\bf q}$ is $0$. 
\end{description}

\begin{enumerate}
\item
\sloppy We will show in the following analysis that it is possible to choose the parameters $\delta$ and $w$, such that $w = O({\left( \frac{1}{\epsilon} \log{\frac{n}{\epsilon}}\right)}^{\frac{c^2}{c^2-1}}{\left(\log{\left(\frac{1}{\epsilon}\log{\frac{n}{\epsilon}}\right)}\right)}^{\frac{3}{2}}(c^2-1)^{-\frac{3}{2}})$ and both properties {\bf P1} and {\bf P2} hold with probability at least $1-\epsilon$. This implies that algorithm {\bf A} succeeds with probability at least $1 - 2\epsilon$. Rescaling $\epsilon$ by $1/2$, we can conclude that algorithm {\bf A} succeeds with probability at least $1 - \epsilon$. 

\sloppy Choose\footnote{Note that $\frac{1- p_2(\delta/c)}{1 -p_1(\delta)}$ is an increasing function of $\delta$ for a fixed $c$ and hence for any $n,\epsilon$,$\exists \delta$ which satisfies this condition} $\delta$ and $w$ such that  
\begin{equation} \begin{array}{l} 
\frac{1- p_2(\delta/c)}{1 -p_1(\delta)}= \frac{c^3}{5}e^{\frac{\delta^2}{2c^2}(c^2-1)}= \frac{1}{\epsilon}\log{\left(\frac{n}{\epsilon}\right)} \\
 w = k \log{(n/\epsilon)}, \text{ where } k = \frac{1}{1-p_2(\delta/c)} \label{choice of delta} 
\end{array}
\end{equation}   
\begin{itemize}
\item The choice of $k$ is such that $p_2(\frac{\delta}{c}) = 1 - \frac{1}{k}$. This implies that the false positive probability with respect to any particular point in $S({\bf q},c)$ is at most ${\left(p_2(\frac{\delta}{c})\right)}^w={\left(1 - \frac{1}{k}\right)}^{k\log{(n/\epsilon)}}\leq \frac{\epsilon}{n}$. Hence the expected number of false positives in the output of Algorithm {\bf A} is at most $\epsilon$. By Markov inequality, the probability that the output of the TCAM is a false positive match is at most $\epsilon$. Hence the property {\bf P2} holds with probability atleast $1-\epsilon$. 

\item Using \eqref{choice of delta}, we get $p_1(\delta) = 1 - \frac{\epsilon}{w}$. 
Hence the probability of making a false negative error on any ternions is at most $\epsilon/w$. Using the union bound implies that probability of a false negative in the output of the TCAM is at most $\epsilon$. i.e. {\bf P1} holds with probability at least $1-\epsilon$. 
\end{itemize}
Now using \eqref{choice of delta}, we get 
\begin{equation}\begin{array}{l}
\frac{\delta^2}{c^2} = \frac{2}{c^2-1}\log{\left(\frac{5}{c^3\epsilon}\log{\left(\frac{n}{\epsilon}\right)}\right)}. \\ 
k = \\
O({\left( \frac{1}{\epsilon}\log{\left(\frac{n}{\epsilon}\right)}\right)}^{\frac{1}{c^2-1}}{\left(\log{\left(\frac{1}{\epsilon}\log{\left(\frac{n}{\epsilon}\right)}\right)}\right)}^{\frac{3}{2}}(c^2-1)^{-\frac{3}{2}}) \\  
w = \\
O({\left( \frac{1}{\epsilon}\log{\left(\frac{n}{\epsilon}\right)}\right)}^{\frac{c^2}{c^2-1}}{\left(\log{\left(\frac{1}{\epsilon}\log{\left(\frac{n}{\epsilon}\right)}\right)}\right)}^{\frac{3}{2}}(c^2-1)^{-\frac{3}{2}}).
\end{array}
\end{equation}  
\item
Using error probability $1/4$ in the analysis of [Theorem 4,1], we get that the algorithm succeeds with probability at least $1/2$ and the width of the TCAM required is given by  $w = O({(\log{n})}^{\frac{c^2}{c^2-1}}{(\log{\log{n}})}^{\frac{3}{2}})$. If this process is repeated $O(\log_2{(1/\epsilon)})$ times, the probability of success can be amplified to $1-\epsilon$.  
\item
Choose $\delta = \alpha c$ and $w = k\log{(n/\epsilon)}$ where $\alpha$ is such that $k = \frac{1}{1-p_2(\alpha)}$. The condition $k \geq \frac{1}{1-p_2(2)}$ ensures that $\alpha \geq 2$ and thus $\delta \geq 2c$.
\begin{itemize}
\item  
Again, the choice of $k$ is such that $p_2(\frac{\delta}{c}) = 1 - \frac{1}{k}$. Repeating the analysis of [Theorem 4, 1] we get that the property {\bf P2} holds with probability at least $1-\epsilon$. 
\item Now $p_1(\delta)= p_1(\alpha c) = 1 - \frac{e^{-2c^2}}{8c^3\sqrt{2\pi}}  \geq 1 - e^{-c^2} $. Now if $c^2 \geq \log{(k/\epsilon)\log(n/\epsilon)}$ i.e.  then we have $p_1(\delta) \geq 1 - \epsilon/w$. Again, similar to [Theorem 4, 1], this implies that {\bf P1} holds with probability atleast $1-\epsilon$. Hence algorithm {\bf A} solves the $(1,c)$- Near Neighbor problem with an error probability of at most $\epsilon$ using a TCAM of width $k\log{n}$ when $c^2 \geq \log{\left(\frac{k}{\epsilon}\log{\frac{n}{\epsilon}}\right)}$.
\end{itemize}
\end{enumerate}
\end{proofof}

\subsection{The $c$-ANNS problem}\label{canns}

\fussy Consider a data set $S$ consisting of $n$ points and a query point ${\bf q}$. Let $r_0$ and $r_{max}$ denote the smallest and largest possible distances from ${\bf q}$ to its nearest neighbor in $S$ and let $m = \lceil 2\log{r_{max}/r_0} \rceil$. To solve the $c$-ANNS problem we use a simple (but weak\footnote{The weakness of this reduction is because of the possibility that $m$ might be large or unbounded. We remark that the approach in \ref{canns} cannot be trivially modified to use the ``adaptive'' reduction of $c$-ANNS to $O(\log{\frac{n}{c-1}})$ instances of $(l,c)$-NN problem proposed by Har-Peled \cite{HP01}}) reduction \cite{IM98,GIM99} from $c$-ANNS to $m$ instances of $(1,\sqrt{c})$-NN problem. Next, we describe the pre-processing step. Let the parameters $\delta,w$ be chosen as in the analysis of [Theorem 4,1] such that the error probability in solving a $(1,\sqrt{c})$-NN problem on $S$ is at most $\epsilon/m$.

\noindent{\bf For each} $i$ in $1 \ldots m$:
\begin{enumerate}
\item Let  $l_i = r_0c^{\frac{i-1}{2}}$, $u_i = r_0c^{\frac{i}{2}}$.
\item Scale down the coordinates of the data points by $l_i$ and find ternary hash representations of the data points using a $(1,\sqrt{c},p_1(\delta),p_2(\delta/\sqrt{c}))$-TLSH family. 
\item Store the hash representations in the TCAM of width $w$, in order of increasing $i$.
\end{enumerate}
The TCAM lookup of the hash representation of ${\bf q}$, i.e. $T({\bf q})$ (using the same hash functions) is output as the $c$-approximate nearest neighbor. Let $l^*$ denote the distance of ${\bf q}$ to its nearest neighbor in $S$, i.e. $l^* = \text{argmin}_{{\bf s} \in S}\norm{{\bf s-q}}$ and $i^*$ denote the first $i$ in $1 \ldots m$ for which $l_i \geq l^*$. Then the correct solution $(l_{i^*},u_{i^*})$-NN problem yields the $c$-approximate nearest 
neighbor of ${\bf q}$. This is because $l^* > l_{i^*-1}$ and the output is at a distance of at most $u_{i^*} = cl_{i^*-1}<cl^*$ from ${\bf q}$. The choice of parameters $\delta$ and $w$ is such that each $(l_i,u_i)$-NN problem is solved with an error probability of at most $\epsilon/m$. Hence the probability of making an error in solving any one of the m the $(l_i,u_i)$-NN problems is at most $\epsilon$. This approach can be generalized to using TCAMs with smaller widths but $O(m\log{(1/\epsilon)})$ lookups per query point in a manner similar to [Theorem 4.1,2]. As mentioned before, Tao \etal \cite{tao} describe a method to solve the $c$-ANNS problem without solving a sequence of near neigbor problem, using the computation of longest common prefixes of binary strings. It would be interesting to find out if their approach can be adapted for use with TCAMs in order to avoid solving a sequence of $(l_i,u_i)$-NN problems.

%If the $(l_{i^*},u_{i^*})$-NN problem is solved correctly, the data point output by the TCAM is a $c$-approximate nearest neighbor of query ${\bf q}$.
%Hence the above algorithm solves the $c$-ANNS problem with error probability at most $\epsilon$. Then the probability of making an error in the output is at most the probability of making an error while solving the $(l_{i^*},u_{i^*})$-NN problem, i.e. $\epsilon$. 

\section{Simulations and Experiments}\label{simu:expt}
In this section we explore the trade-off between the width of the TCAM and the performance of the algorithm {\bf A}. In particular we show via simulations that a TCAM of width $288$ bits solves the $(1,2)$-NN problem on practical and artificially generated (but illustrative) data sets consisting of $1$M points in $64$ dimensional Euclidean space. Finally, we also design an experiment with TCAMs inside an enterprise ethernet switch (Cisco Catalyst 4500) to show that TLSH can be used to configure a TCAM to perform 1.488 million queries per second per 1Gbps port.

\subsection{Simulations:}\label{simu}
\begin{figure*}\label{veri}
\centering
\subfigure[Fscore]{\includegraphics[scale=0.32]{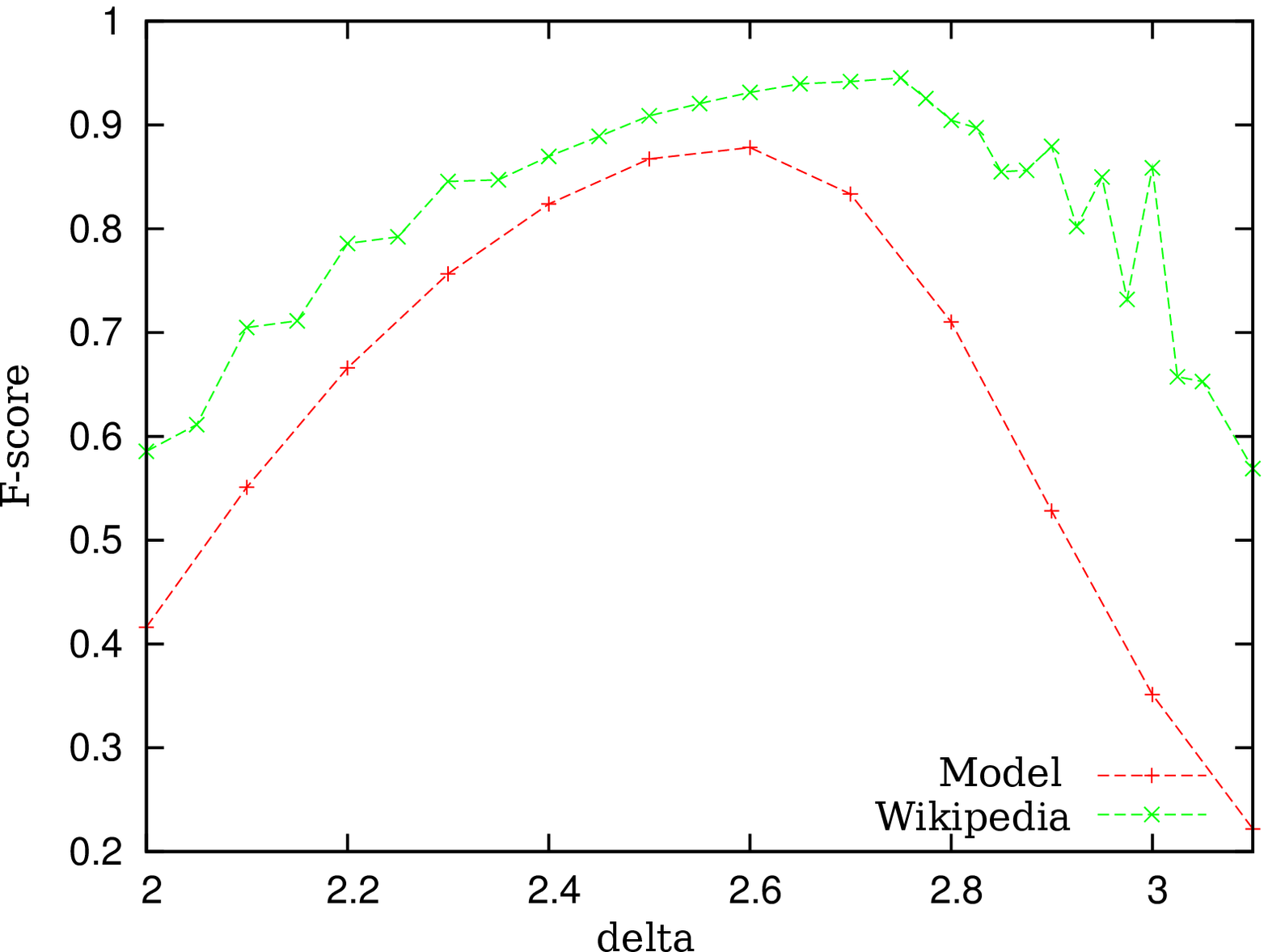}}
\subfigure[False positives per query]{\includegraphics[scale=0.32]{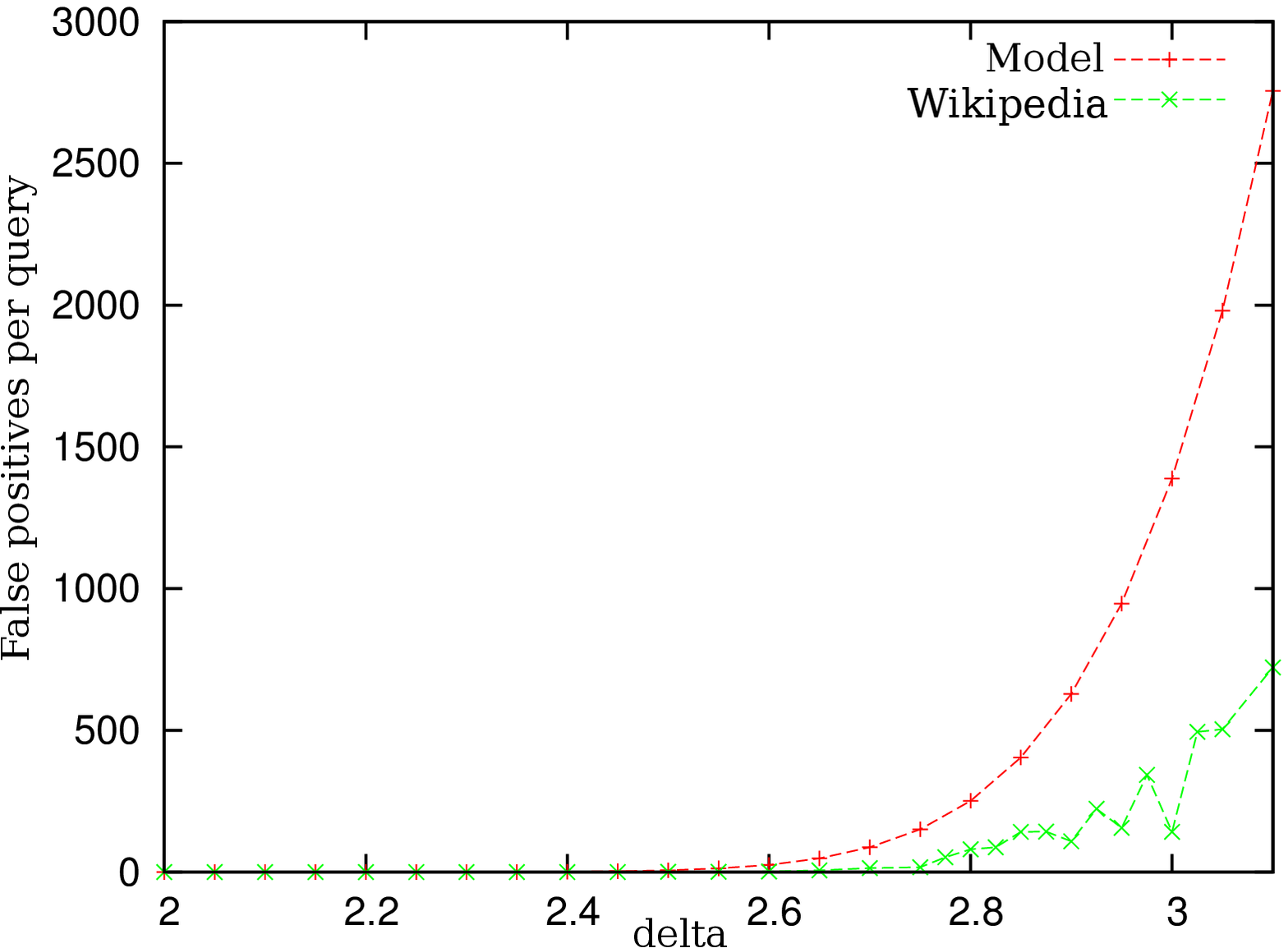}}
\subfigure[False negative rate]{\includegraphics[scale=0.32]{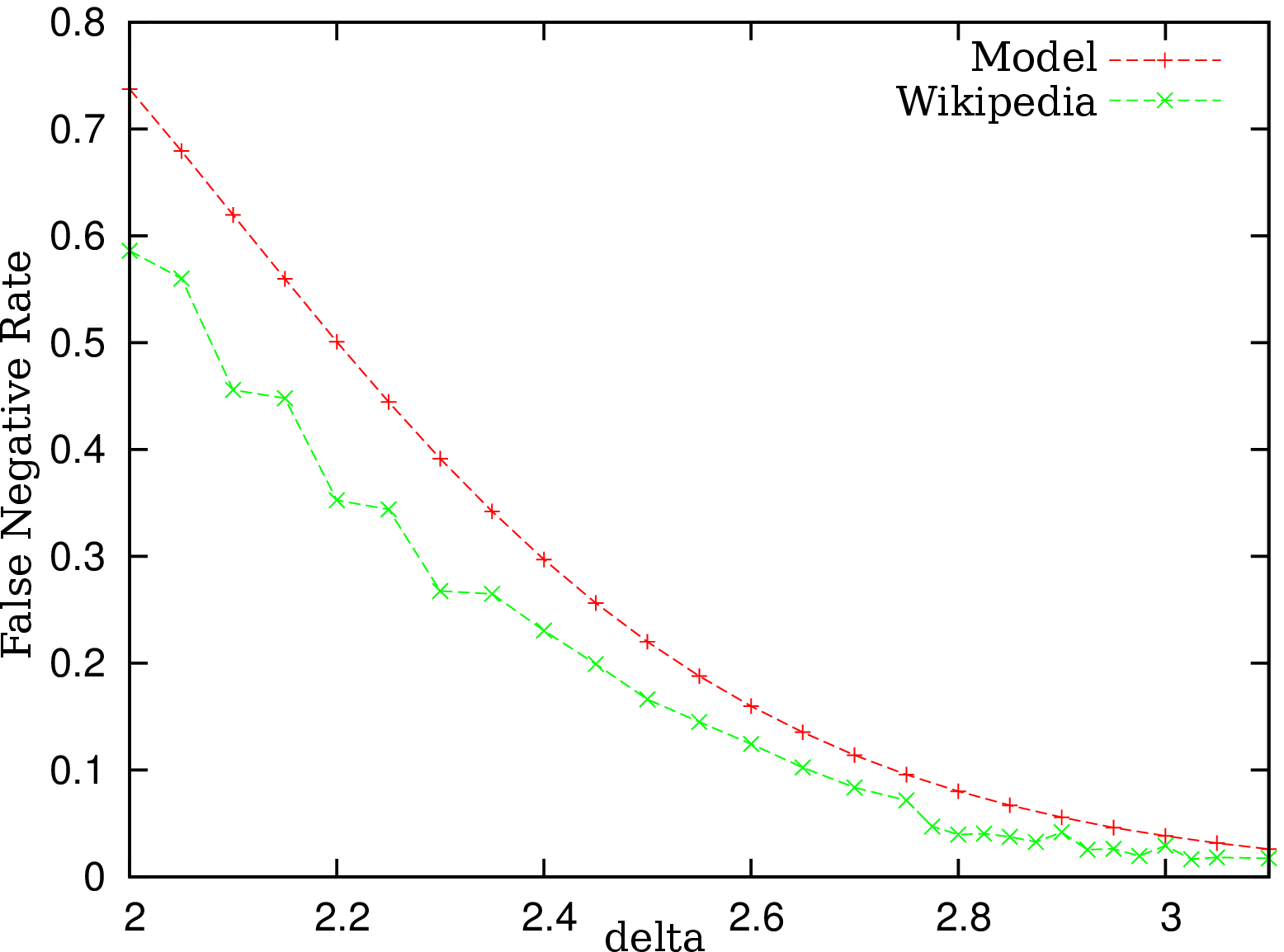}}
\caption{Variation of performance measures observed for the Wikipedia data set with $\delta$ and the corresponding model predictions as described in section \ref{model}, using a TCAM of width $w = 288$ bits. As we can see in this figure, increasing $\delta$ decreases the false negatives but increases the false positives. Thus there exists an optimal choice of $\delta$ which minimizes false positives when false negative rate is below some threshold and maximizes the F-score. This figure also shows the comparison between the Wikipedia data set and its corresponding model, using a TCAM of width $w = 288$ bits.}
\end{figure*}

We evaluate our algorithm on 3 specific data sets with query points generated artificially. Each data set contains a million points chosen from a $64$ dimensional Euclidean space ($n = 10^6$, $d = 64$). The corresponding query set contains $1$K points generated from a $64$ dimensional Euclidean space. We list the data sets we used ordered from the most "benign" to the "hardest" as follows.
\begin{description}
\item["Random" data:] 
\sloppy We chose data points generated uniformly at random from the $d$-dimensional cube $C_d = {\{-2/\sqrt{d},2/\sqrt{d}\}}^d$ for this data set. We chose half of the query points by selecting a random data point ${\bf s}$ (say) and choosing a point uniformly at random from the surface of a sphere of radius $l$ centered at ${\bf s}$. We generated remaining half of the query points uniformly at random from $C_d$. The size of the cube as well as the choice of the query points ensured that a significant fraction of the query point - data point pairs were separated by a distance either at most $l$ or at least $cl$. (Note that query point - data point pairs such that distance between them lies in between $l$ and $cl$ do not contribute to either false positives or false negatives and can safely be ignored. )

\item[Wikipedia data set:]
\fussy The second data set we used is the semantically annotated snapshot of the English Wikipedia (SW v.2) data set, obtained from Yahoo!. It contained a snapshot of the English Wikipedia (from 2005) processed with publicly available NLP tools. We computed the "simHash" signatures \cite{C02,MJS07}, and embedded the signatures in a Euclidean space\footnote{We used an appropriate scaling in order to ensure that a significant fraction of the query point - data point pairs are such that the distance between the query point and the data point was either at most $l$ or at least $cl$}. The query points were generated by randomly choosing $1$K data points of the data set and flipping a few (at most $3$)\footnote{The perturbation was chosen according to the experimental study of near duplicate detection in web documents\cite{MJS07}.} randomly chosen bits of their simHash signatures. 

\item["Threshold" data:] 
The third data set we used was artificially designed to maximize the number of false positives and false negatives. 
A single query point was generated uniformly at random from $C_d$. In order maximize the number of false negatives and the number of false positives seen, half of the data points were chosen to lie on the surface of a sphere $S_1(\bf{q})$ (say) of radius $l$ centered at $q$ and the remaining half are chosen to lie on the surface of a sphere $S_2(\bf{q})$ (say) of radius $cl$ centered at $q$ (The data points were on the "threshold" of being similar and dissimilar to $q$). This setup was repeated for each of the $1$K query points and the average values of the false negatives and false positives  observed are reported. 
\end{description}

Apart from presenting the number of false positives observed per query and the false negative rate (fraction of false negatives observed) as a measure of accuracy, we also report the F-score  or the $\text{F}_1$-measure \cite{RIJ79} of our algorithm which is just the harmonic mean of precision and recall. Precision is defined as the fraction of retrieved documents that are relevant. Recall is the fraction of relevant documents that are retrieved. Similar to precision and recall, the Fscore lies in the range $[0,1]$ and a intuitively a high value of F-score implies high values of precision and recall. 

\sloppy To explore the trade-off between accuracy and TCAM width, we choose the TCAM widths in the range $w$ = $32,64,96,128,144,160,192,224,256,288,320$ bits. (Note that commercially available TCAMs have $72$,$144$,$288$ bit configurations). As we are interested in an accurate algorithm, as a design choice we set the the tolerance of the false negative rate at $\epsilon_n = 5\%$ and minimize the number of false positives generated under this constraint. For each $w$, we choose $\delta$ for which the least number of false positives are observed while ensuring that the false negative rate is below $5\%$. For F-score, we chose the $\delta$ which maximizes the F-score using a binary search. We illustrate this procedure for a TCAM of width $w = 288$ bits as shown in Figure \ref{veri}. As expected, increasing $\delta$ decreases the false negative rate but increases the number of false positives and thus generates a bell shaped curve for the Fscore. The figure shows that there exists an optimal choice of $\delta$ which minimizes the false positives or maximizes the Fscore. We refer to this choice of $\delta$ as $\delta_{\text{opt}}$. 

 \subsection{Model}\label{model}
\fussy The process just described for arriving at the optimal choice of $\delta$ involves the use of the query points. Hence, the optimal value of $\delta$ can not be precomputed given just the database. However, it turns out that only an estimate of the distribution of query points gives a good approximation to choosing the optimal $\delta$. Let $n_1$ denote an estimate of the no. of data points which are "similar" to the query. Let $n_2$ denote an estimate of the number of data points "dissimilar" to the query. Consider a model containing a single query point $q$ with $n_1$ points on $S_1(q)$ and $n_2$ points on $S_2(q)$. Then for a TCAM of width $w$ using the expressions for $p_1(\delta)$ and $p_2(\delta/c)$ it is possible to theoretically calculate the expected values of the false negative rate, no. of false positives per query and the expected f-score for this model and use them as a predictions for choosing $\delta$. For each data set, we use the average number of similar and dissimilar points to a single query (by averaging over the $1$K queries) as $n_1$ and $n_2$ in the model. The observed values of the false negative rate, number of false positives per query, and the f-score as $\delta$ is varied were found to closely match those predicted by the model. For example, a comparison of these quantities observed for the Wikipedia data set as $\delta$ is varied with those predicted by a model for this data set is shown in figure \ref{veri}.  

\subsection{Results and discussion} 
\begin{figure*} \label{results_fscore}
\centering
\subfigure["Random" data set] {\includegraphics[scale=0.32]{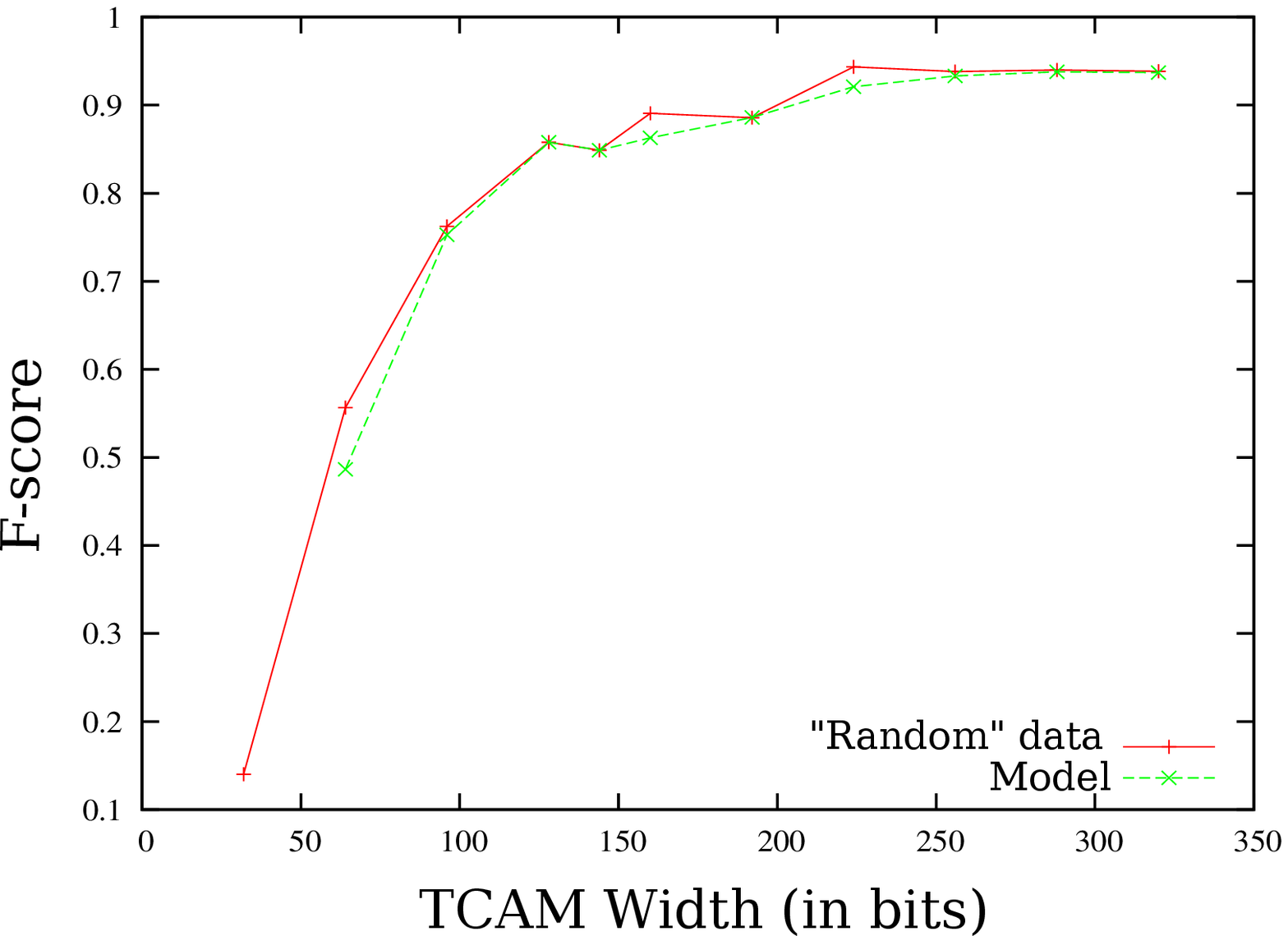}}
\subfigure[Wikipedia data set] {\includegraphics[scale=0.32]{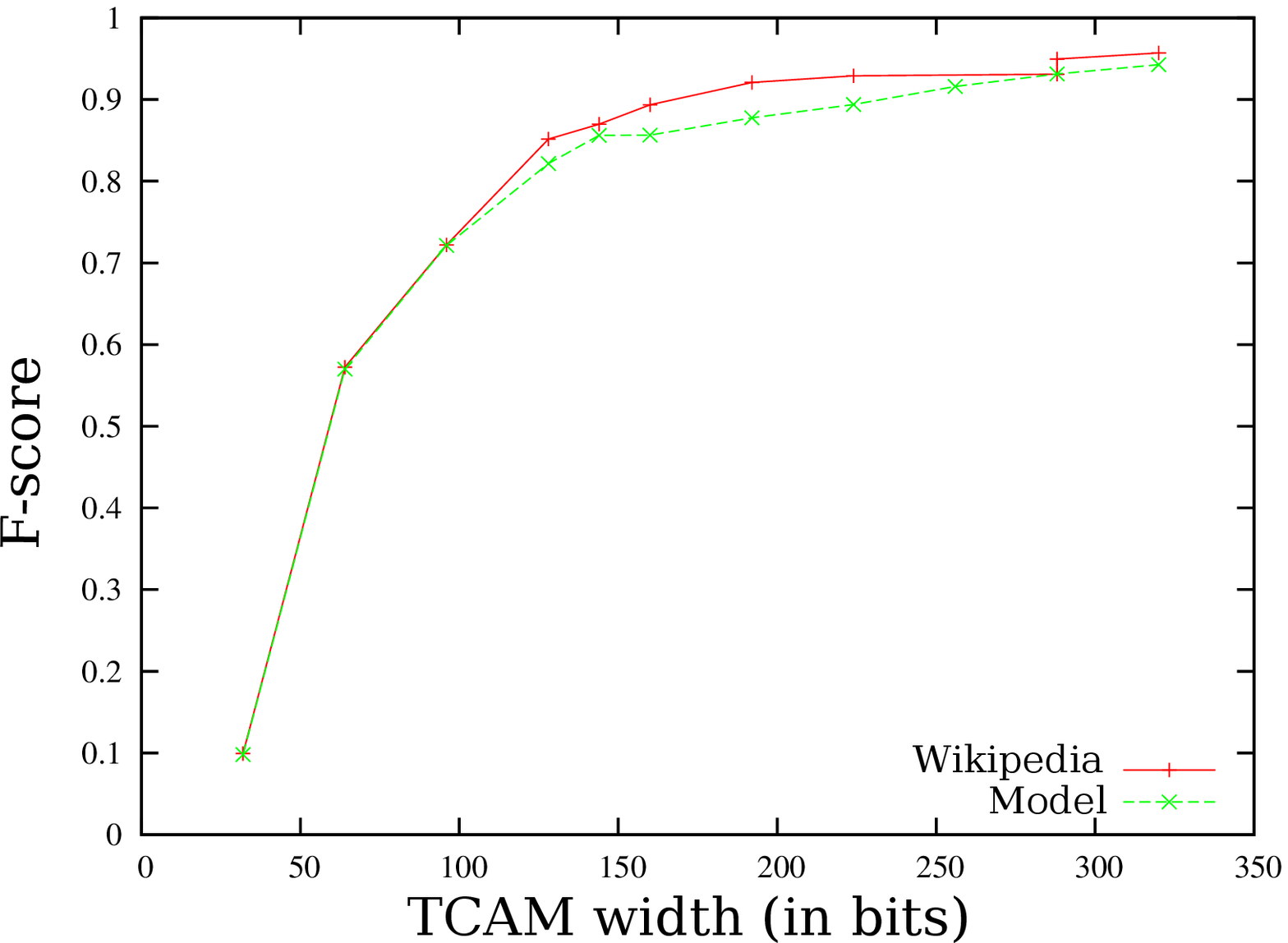}}
\subfigure["Threshold" data set] {\includegraphics[scale=0.32]{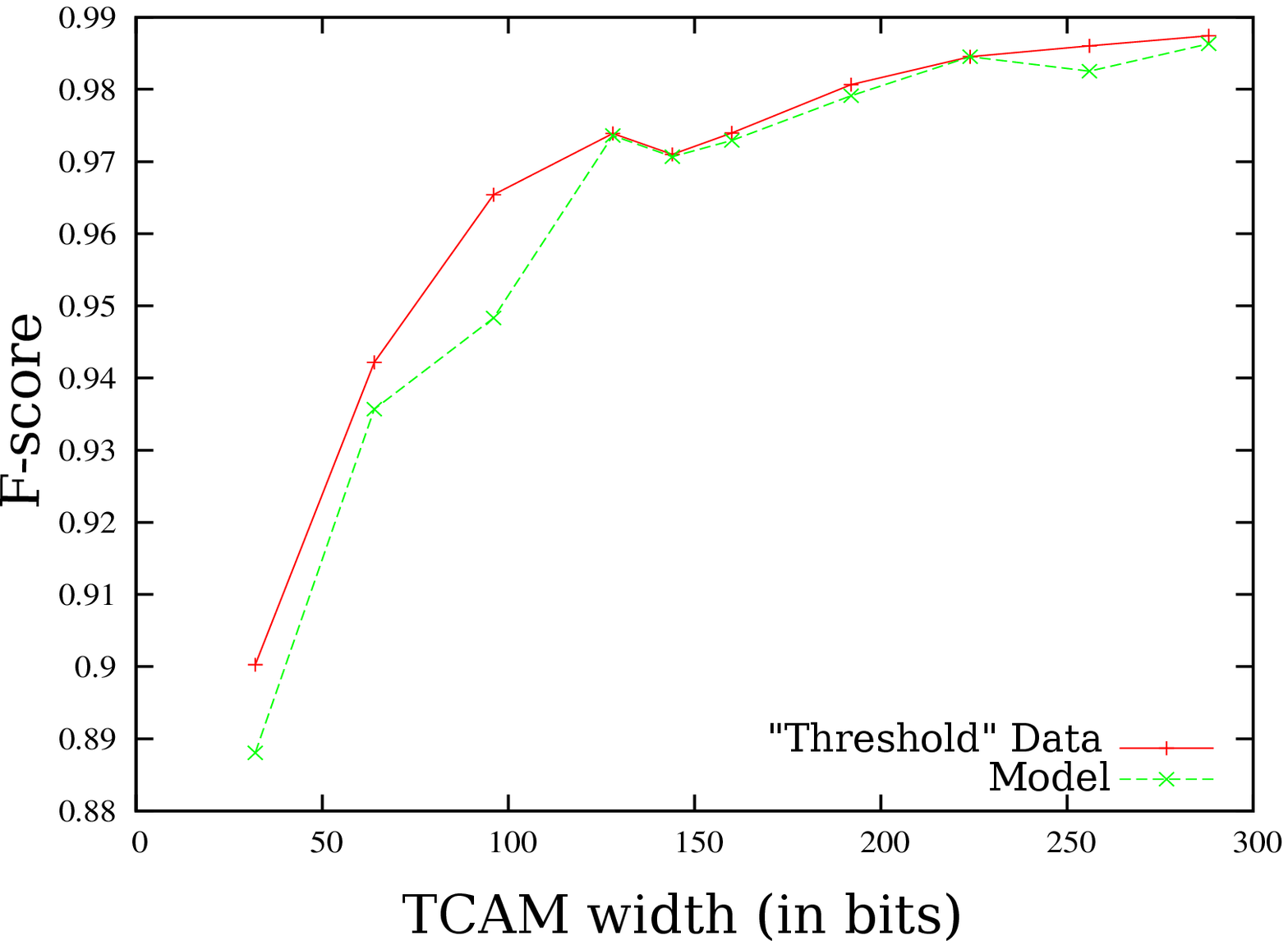}}
\caption{Variation of the F-score vs the width of the TCAM: As we can observe from this figure, F-score of our algorithm increases with the width of the TCAM used for different data sets. This figure also shows that on range of data sets, use of a TCAM of width $288$ bits results in a method with the F-score approximately $0.95$. Finally, this figure also shows there is only a slight loss in performance if $\delta$ is precomputed according to the model, as opposed to being chosen optimally.}
\end{figure*}

\begin{figure*} \label{results_fpfn}
\centering
\subfigure["Random" data set] {\includegraphics[scale=0.32]{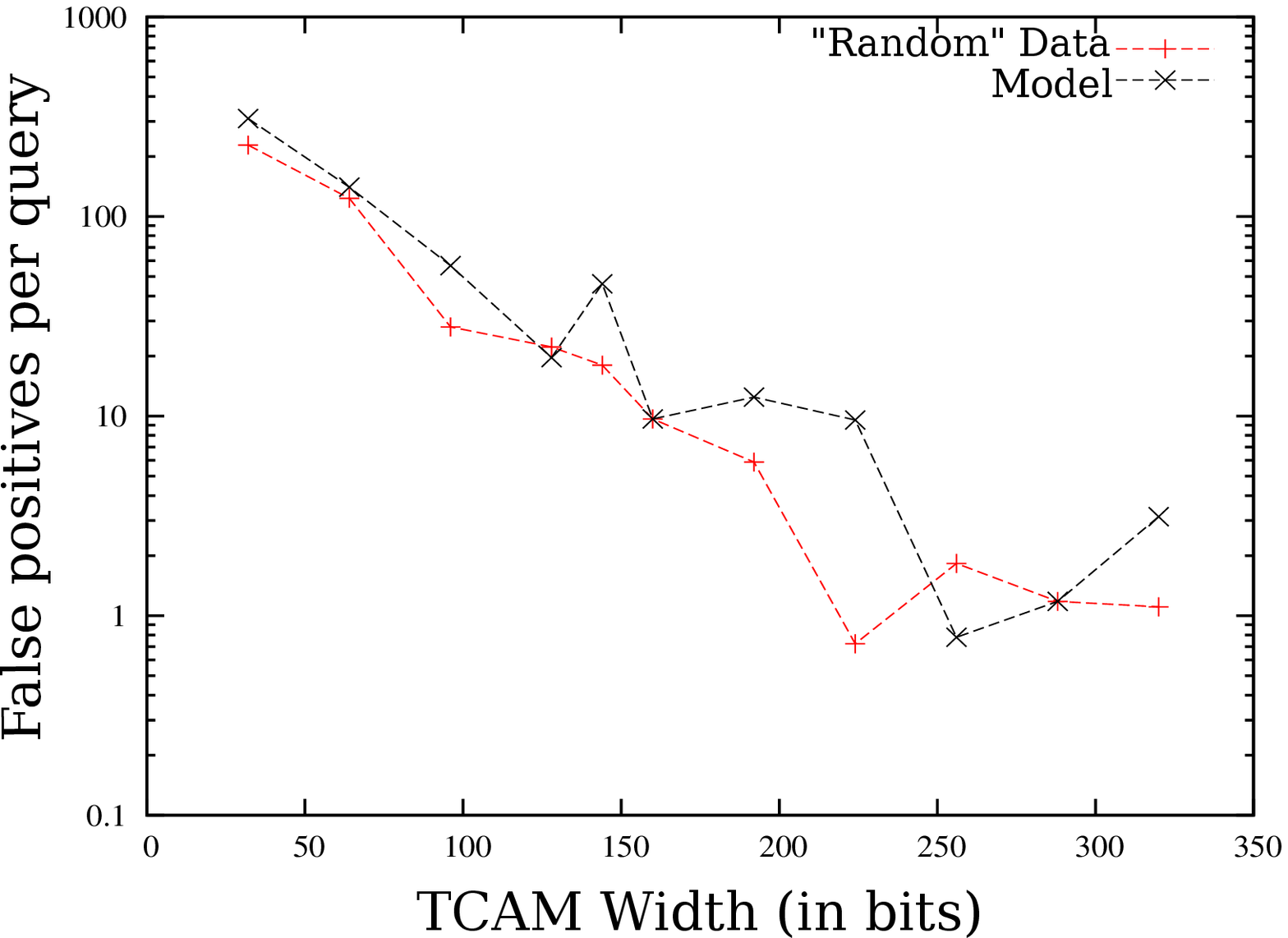}}
\subfigure[Wikipedia data set] {\includegraphics[scale=0.32]{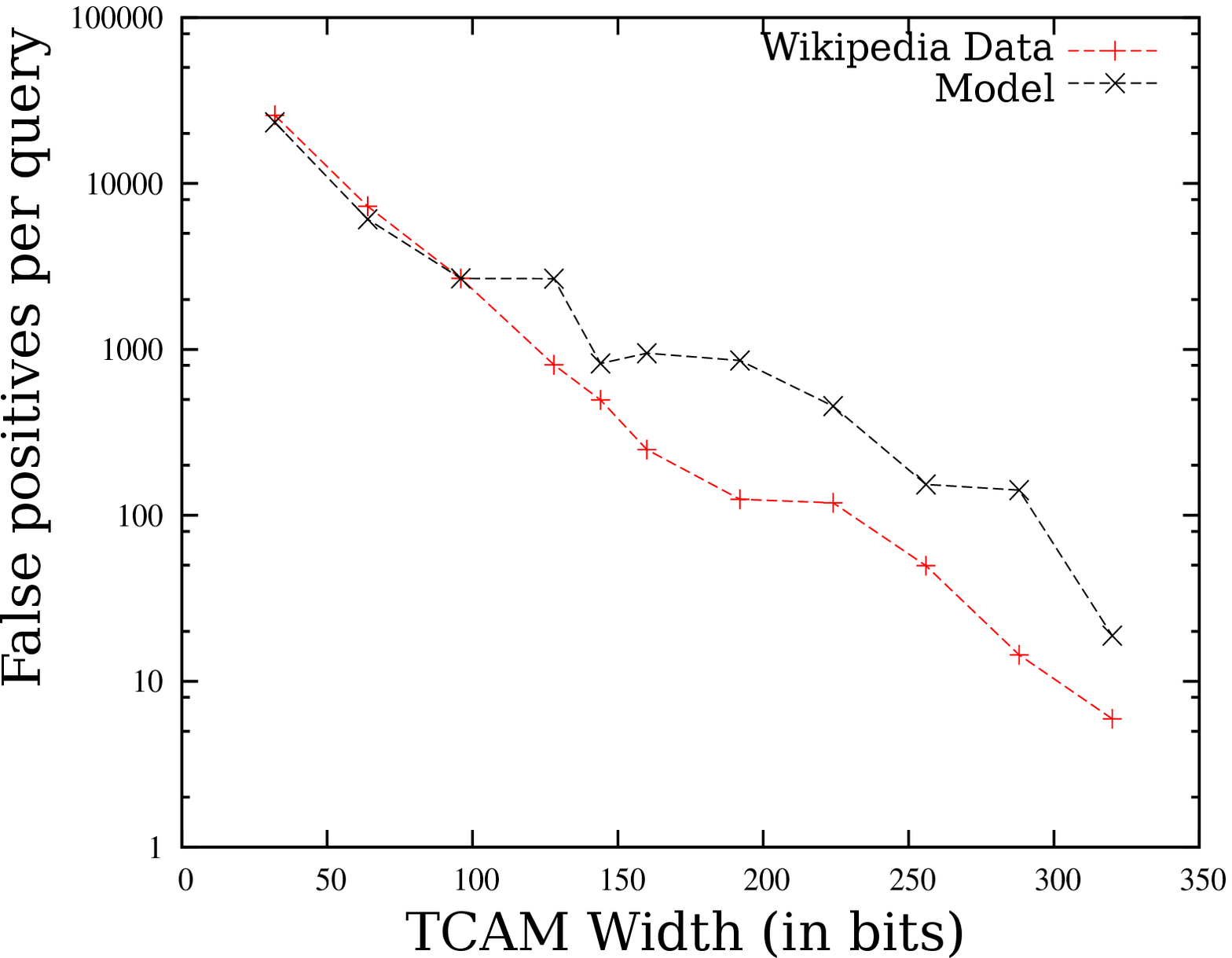}}
\subfigure["Threshold" data set] {\includegraphics[scale=0.32]{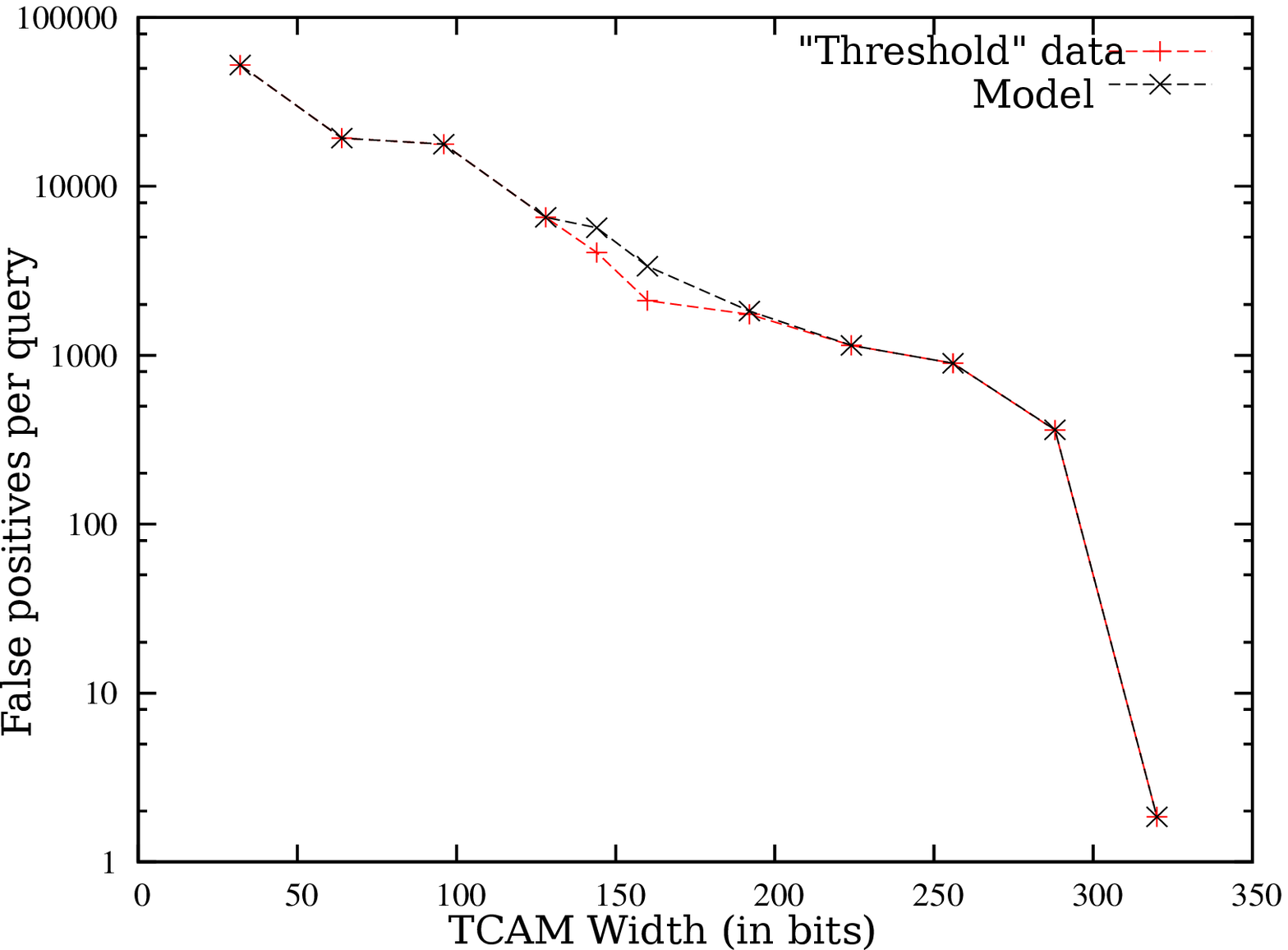}}
\caption{Number of false positives per query vs the width of the TCAM with false negative rate capped at $5\%$. This figure shows that the number of false positives per query drops rapidly as the width of the TCAM is increased with the false negative rate capped at $5\%$. This suggests that on practical data sets, use of a TCAM of width $288$ bits generates few false positives}
\end{figure*}

We observe that performance of our algorithm i.e. the F-score and the number of false positives generated, improves as the width (size) of the TCAM is increased as seen in figures \ref{results_fpfn} and \ref{results_fscore}. As seen in the figure, the improvements in the F-score follow the law of diminishing returns for increasing TCAM widths and a F-score better than 0.95 is obtained using a TCAM of width $288$ bits for all the data sets considered which intuitively indicates high values of precision and recall. Secondly, we note that while tolerating a false negative rate of 5\%, only $1$ false positive was observed per query for the "Random" data set. For the Wikipedia data set, the number of false positives observed per query was $14$ while for the threshold dataset $51$ false positives were observed, while the false negative rate was below the threshold of $\epsilon_n = 5\%$.  These simulation results suggest that a TCAM of width $288$ bits can be used to solve the $1,2$-NN problem on data sets consisting of a million points. 

We seek solutions in which false negative rate is at most 5\% and the number of false positives generated per query by the method is at most $10$. For the "Random" data set, the use of a $288$ bit TCAM actually satisfies these demands, while for the Wikipedia data set, the use of a $288$ bit TCAM comes very close to matching these requirements. Since $288$ bit wide TCAMs containing 0.5M entries are available in the market, our method represents a novel yet easy solution to the problem of similarity search in high dimensions. Even though a larger number of false positives are generated (51) by using a 288 bit wide TCAM on the "Threshold" data set, we note that this data set was artificially constructed to maximize the number of false positives and false negatives and we conjecture the property of all the similar points to a query being on the "threshold" of being similar and dissimilar points being on the "threshold" of being dissimilar is unlikely to be observed in practical data sets. We would also like to mention here that it is also possible to generate a worst case input distribution for the F-score which has just a single point similar to a given query point (on the sphere $S_1(q)$) and all the remaining data points are dissimilar to the query (on the sphere $S_2(q)$). Running the simulations on this data set we observed that the performance was not too worse than the results presented in this section, even though this property (of having a single similar data point to a query) is unlikely to be observed in real data sets. 

\subsection{Preliminary experimental validation}

In this section, we demonstrate that the simulations of TLSH are
realistic, and that the TLSH algorithm can be made to work with existing TCAM
based products at very high speeds. For this, we need to choose an appropriate
platform. Although it is possible to use a standalone TCAM platform,
managing the TCAM in software is non trivial. For a preliminary
validation, we leverage a Cisco Catalyst 4500 (Cat4K) series enterprise
switch \cite{cisco4500} which uses TCAMs for a variety of purposes
including implementing access control lists (ACLs). In one second, it
can support up to a billion TCAM lookups and switch 250 million
packets.

Our simple observation is as follows. For validating a $64$-bit TCAM
lookup, we map it to an IP address lookup in a $64$-bit IPv4
access control list. For example, a $64$-bit lookup key could be
represented as a $32$-bit IPv4 source and a $32$-bit IPv4 destination
address. This query is embedded within an IPv4 source and destination
address fields of an IP packet and injected into the Cisco switch.
Access control lists involve TCAM lookups. The TCAM database is
similarly represented as entries of an ACL with permit action for
matches, i.e.  if the TCAM matches a given query, the action
would be to permit the IP packet and if there is no match, the action
would be to drop the packet. Thus, all egress packets represent
queries that had a TCAM hit as shown in figure \ref{ciscofigure}.

We use a high speed commercial traffic generator (from IXIA). Though the
Cat4K switch can support up to 384 1Gb/s ports, we use two 1Gb/s
ports for this experiment, and connect these to two ports of IXIA,
which are programmable and can inject traffic with specified IP
addresses. We pass packets from one port and detect egress packets on
the other via the switch. A switch learns the source and the
destination for the given hardware MAC addresses of a packet (that we
set manually) and switches these packets in hardware.  We inspect the
egress packets' IP addresses to determine which queries hit the TCAM.
To ensure the speed, we send IP packets (representing queries) at wire
speed (i.e. ~1.5 million packets per 1Gb/s port).

We validated several randomly generated data sets, for $32$ and $64$
bit TLSH lookups. For each data set, we randomly generate negative,
positive and false positive queries and the inspect the egress
packets' IP addresses. We observe that for every positive or false
positive query (according to TLSH), we do indeed have an egress packet
with the corresponding IP address. For every negative query, we never
detect the corresponding IP packet at egress. We believe that this
simple experimental setup is novel as it allows us to rapidly
demonstrate the performance argument without the overheads of managing
TCAMs!

\begin{figure*}\label{ciscofigure}
\centering
\includegraphics[scale=0.8]{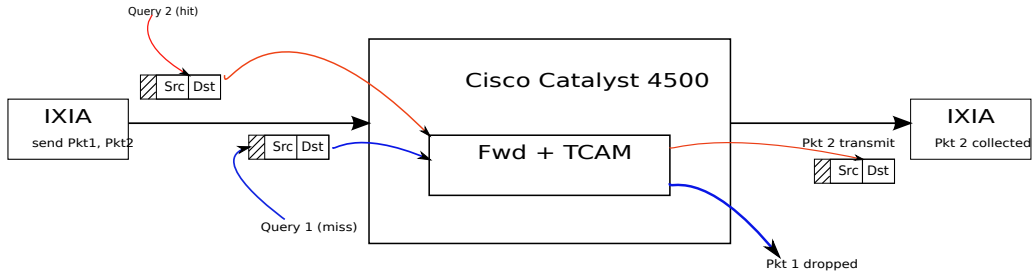}
\caption{Block diagram of the experimental setup. Queries are inserted in the ipv4 header and pumped into the switch using a IXIA traffic generator. ACLs are programmed using the TLSH algorithm and are applied to ingress packets. For a ACL match we forward the packet and drop otherwise. Egress packets are collected at another IXIA port and they correspond to the matched queries}
\end{figure*}

\section{Related work}\label{relwork}
Early methods to solve similarity search problems in high dimensions used the space partitioning approach in order to solve the exact nearest neighbor problem by reducing the candidate set of data points for a given query, using branch and bound techniques. They includes the famous k-d tree approach \cite{Bentley:kdtree}, cover trees \cite{Kakade:covertrees}, navigating nets \cite{Krauthgamer:lee:navigating:nets}. However an experimental study \cite{expt:space:partitioning} has showed that approaches based on space partitioning scale poorly with the number of dimensions $d$ and in fact when $d > 10$, they performed worse than a brute force linear scan for some specific data sets (curse of dimensionality).
 
%In order to overcome the curse of dimensionality, Indyk and Motwani \cite{IM 98} proposed the concept of locality sensitive hashing (LSH) which solved teh $c$-approximate nearest neighbor problem with space requirement and query time polynomial in the size of the database and the number of dimensions. 
%One of the first algorithms proposed to solve the $(l,c)$-Near Neighbor problem was proposed by Kushilevitz \etal \cite{KOR98}. It had a query time of poly$(d,\log{n},{(c-1)}^{-1})$ and space requirement of $n^{O(1/{(c-1)}^2)}$. It was based on the technique of dimensionality reduction proposed by Johnson and Lindenstrauss \cite{JL84}. Further work \cite{CR04,AC06} improved the query time to $(d+\log{n}+{(c-1)}^{-1})^{O(1)}$ using the same space requirement $n^{O(1/{(c-1)}^2)}$. Andoni \etal \cite{AIP06} recently proved a lower bound of $n^{\Omega({1/(c-1)}^2)}$  on the space required by any data structure which solves the $(1,c)$-Near Neighbor problem  with a constant number of probes. Hence the aforementioned algorithms are space-optimal. For values of $c$ close to $1$, this leads to a large space requirement which seems to be impractical for applications.  
 
Locality sensitive hash (LSH) family was proposed by Indyk and Motwani \cite{IM98} to solve the $c$-ANN problem with space requirement and query time polynomial in the size of the database and the number of dimensions. Given parameters $l$, $u$, $p_l$ and $p_u$, a $(l,u,p_l,p_u)$-LSH family of hash functions has the following property: The probability that two points separated by a distance atmost $l$ are hashed to the same value is at least $p_l$ and probability that two points separated by a distance at least $u$ are hashed to the same value is at most $p_u$. Gionis \etal \cite{GIM99} showed a framework based on a $(l,u,p_l,p_u)$-LSH family (where $u = cl$), to solve the $(l,c)$-Near Neighbor problem in time $O(dn^\rho\log{n})$ using space $O(dn + n^{1+\rho}\log{n})$ where $\rho = \log{p_l}/\log{p_u}$. Their algorithm used a LSH family with $\rho = 1/c$. For the case of Euclidean space, the exponent $1/c$ was improved to $\beta/c$ for some fixed constant $\beta < 1$ by Datar \etal \cite{DIIM04}. A near linear storage space solution was proposed by Panigrahy \cite{P06} which has space requirement of $\tilde{O}(n)$ and but a larger query time $\tilde{O}(n^{2.09/c})$ using entropy based techniques along with using the LSH family. Building on this work, Lv \etal \cite{Charikar:multiprobe} suggested the use of multi-probe LSH methods to reduce the number of hash tables required for solving the $c$-approximate nearest neighbor problem \cite{Charikar:multiprobe}. Andoni and Indyk \cite{AI06} further improved the value of $\rho$ (for Euclidean space) to $1/c^2+o(1)$. This value of $\rho$ is near-optimal since it matches the lower bound for LSH proved by Motwani \etal \cite{MNP06}. 

For $c \approx 1$, the near quadratic space requirement of the {\emph optimal} LSH could be a hindrance in solving large problems like image similarity with millions of images in the data set\cite{WTF08:NIPS}. In fact recent studies have shown that machine learning techniques like restricted Boltzmann machines and boosting, out perform LSH when the number of bits available is small and fixed \cite{SH09,WTF08:CVPR}. Also the query time of $O(dn^{1/c^2})$ makes the application of LSH for proximity based methods like clustering and classification difficult in a streaming environment. Hence, in this paper, we consider the use hardware primitives like TCAMs in order to formulate fast, space efficient and accurate methods to solve similarity search problems. 

While TCAMs have been used previously in order to obtain efficient solutions to the problem of finding frequent elements in data streams\cite{divi}, we are not aware of any other work which uses TCAMs for solving similarity search and nearest neighbor problems.  
 
In parallel, there has been significant progress in proving lower bounds for the approximate nearest neighbor problem using the cell probe model \cite{C02,AIP06,PTW08,Yao:cellprobe}. In particular Panigrahy, \etal \cite{PTW08} show that a data structure which solves the $c$-ANNS problem using $t$ probes must use space $n^{1+\Omega(1/(c^2t))}$. This implies that any data structure that uses $\tilde{O}(n)$ space with poly-logarithmic word size, and with constant probability, gives a constant approximation to nearest neighbor problem must be probed $\Omega(\log{n}/\log{\log{n}})$ times. We note that the use of hardware primitives like TCAMs which implement highly parallel operations (not conforming to the cell probe model of computation) enables us to circumvent these lower bounds.
\vfill\eject

\section{Conclusion}\label{conclusion}
In this paper we have proposed a new method to solve the approximate nearest neighbor problem which yields an exponential improvement over existing methods. This improvement is brought about by using a hashing scheme which does not conform to lower bounds for standard binary hashing schemes. This hashing scheme (TLSH) is supported by a TCAM. In fact using a TCAM of width poly-logarithmic in the size of the database, the approximate nearest neighbor problem can be solved in a single TCAM lookup. Using simulations we have shown that off the shelf TCAMs with width $288$ bits can be used to solve similarity search problems on various databases containing a million points in $64$ dimensional Euclidean space. We also design an experiment to demonstrate that even existing TCAMs within enterprise ethernet switches can perform 1.5M ANN queries per 1Gbps port. Thus, we believe that TCAM based similarity search might open new vistas in ultra high speed data mining and learning applications.

\section*{Acknowledgements}
We thank Cisco Systems for experimental resources, Yahoo! for providing the English Wikipedia data set, Piotr Indyk for pointing out the reduction from $(l,c)$-NN problem under $l_{\infty}$ norm to the partial match problem, Srinivasan Venkatachary and Sudipto Guha for helpful discussions. We also thank anonymous reviewers for helpful comments and pointing us to \cite{tao}.
\bibliographystyle{abbrv}
\bibliography{references} 

\begin{thebibliography}{10}

\bibitem{AS72}
M.~Abramowitz and I.~Stegun.
\newblock {\em Handbook of Mathematical Functions with Formulas, Graphs, and
  Mathematical Tables, Chapter 7.}
\newblock Dover, 1972.

\bibitem{JHan:stream:class}
C.~Aggarwal, J.~Han, J.~Wang, and P.~Yu.
\newblock On demand classification of streams.
\newblock In {\em Proceedings of the Tenth ACM SIGKDD International Conference
  on Knowledge Discovery and Data Mining}, 2004.

\bibitem{AC06}
N.~Ailon and B.~Chazelle.
\newblock Approximate nearest neighbors and the fast johnson-lindenstrauss
  transform.
\newblock In {\em Proceedings of the 38th Annual ACM Symposium on Theory of
  Computing}, 2006.

\bibitem{indyk:survey}
A.~Andoni and P.~Indyk.
\newblock Near-optimal hashing algorithms for approximate nearest neighbor in
  high dimensions.
\newblock {\em Communications of the ACM}, 51(1):117--122, 2006.

\bibitem{AI06}
A.~Andoni and P.~Indyk.
\newblock Near optimal hashing algorithms for approximate nearest neighbor in
  high dimensions.
\newblock In {\em Proceedings of 47th Annual IEEE Symposium on Foundations of
  Computer Science}, 2006.

\bibitem{AIP06}
A.~Andoni, P.~Indyk, and M.~Patrascu.
\newblock On the optimality of dimension reduction method.
\newblock In {\em Proceedings of the 47th Annual IEEE Symposium on Foundations
  of Computer Science}, 2006.

\bibitem{divi}
N.~Bandi, A.~Metwally, D.~Agrawal, and A.~E. Abbadi.
\newblock Tcam-conscious algorithms for data streams.
\newblock In {\em Proceedings of the 23rd International Conference on Data
  Engineering}, 2007.

\bibitem{Bentley:kdtree}
J.~Bentley.
\newblock Multidimensional binary search trees used for associative searching.
\newblock {\em Communications of the ACM}, 18.

\bibitem{berkhin}
P.~Berkhin.
\newblock {\em A Survey of Clustering Data Mining Techniques}.
\newblock Springer, 2002.

\bibitem{Kakade:covertrees}
A.~Beygelzimer, S.~Kakade, and J.~Langford.
\newblock Cover trees for nearest neighbours.
\newblock In {\em Proceedings of the Twenty-Third International Conference on
  Machine Learning}, 2006.

\bibitem{Buhler}
J.~Buhler.
\newblock Efficient large scale sequence comparison by locality-sensitive
  hashing.
\newblock {\em Bioinformatics}, 17.

\bibitem{C02}
M.~Charikar.
\newblock Similarity estimation techniques from rounding algorithms.
\newblock In {\em Proceedings of 34th Annual ACM Symposium on Theory of
  Computing}, 2002.

\bibitem{cisco4500}
CISCO.
\newblock Cisco catalyst 4500 series http://www.cisco.com/.

\bibitem{google:video:lsh}
M.~Covell and S.~Baluja.
\newblock Lsh banding for large-scale retrieval with memory and recall
  constraints.
\newblock In {\em Proceedings of the International Conference on Acoustics,
  Speech, and Signal Processing}, 2009.

\bibitem{cover67}
T.~Cover and P.~Hart.
\newblock Nearest neighbour pattern classification.
\newblock {\em IEEE Transactions on Information Theory}, 13.

\bibitem{DIIM04}
M.~Datar, N.~Immorlica, P.~Indyk, and V.~Mirrokni.
\newblock Locality sensitive hashing scheme based on p-stable distributions.
\newblock In {\em Proceedings of the 20th Annual ACM Symposium on Computational
  Geometry}, 2004.

\bibitem{Dutta:Cheng}
D.~Dutta and T.~Chen.
\newblock Speeding up tandem mass spectrometry database search: metric
  embeddings and fast near neighbor search.
\newblock {\em Bioinformatics}.

\bibitem{GIM99}
A.~Gionis, P.~Indyk, and R.~Motwani.
\newblock Similarity search in high dimensions via hashing.
\newblock In {\em Proceedings of the 25th VLDB conference}, 1999.

\bibitem{HP01}
S.~Har-Peled.
\newblock A replacement for voronoi diagrams of near linear size.
\newblock In {\em Proceedings of the 42nd Annual IEEE Symposium on the
  Foundations of Computer Science}, 2001.

\bibitem{indyk:clustering}
T.~H. Haveliwala, A.~Gionis, and P.~Indyk.
\newblock Scalable techniques for clustering the web.
\newblock In {\em Proceedings of the Third International Workshop on the Web
  and Databases}, 2000.

\bibitem{indyk:timeseries}
P.~Indyk, N.~Koudas, and S.~Muthukrishnan.
\newblock Identifying representative trends in massive time series data sets
  using sketches.
\newblock In {\em Proceedings of 26th International Conference on Very Large
  Data Bases}, 2000.

\bibitem{IM98}
P.~Indyk and R.~Motwani.
\newblock Approximate nearest neighbors: Towards removing the curse of
  dimensionality.
\newblock In {\em Proceedings of the 30th Annual ACM Symposium on Theory of
  Computing}, 1998.

\bibitem{Krauthgamer:lee:navigating:nets}
R.~Krauthgamer and J.~Lee.
\newblock Navigating nets:simple algorithms for proximity search.
\newblock In {\em Proceedings of the fifteenth annual ACM-SIAM Symposium of
  Discrete Algorithms}, 2004.

\bibitem{Brian:Kulis}
B.~Kulis and K.~Grauman.
\newblock Kernelized locality-sensitive hashing for scalable image search.
\newblock In {\em Proceedings of the Twelth International Conference on
  Computer Vision}, 2009.

\bibitem{38}
K.~Lakshminarayanan, A.~Rangarajan, and S.~Venkatachary.
\newblock Algorithms for advanced packet classification using ternary cams.
\newblock In {\em Proceedings of the ACM SIGCOMM 2005 Conference on
  Applications, Technologies, Architectures, and Protocols for Computer
  Communications}, 2005.

\bibitem{Charikar:multiprobe}
Q.~Lv, W.~Josephson, Z.~Wang, M.~Charikar, and K.~Li.
\newblock Multi-probe lsh: Efficient indexing for high-dimensional similarity
  search.
\newblock In {\em Proceedings of the Thirty-Third International Conference on
  Very Large Data Bases}, 2007.

\bibitem{MJS07}
G.~Manku, A.~Jain, and A.~Sharma.
\newblock Detecting near duplicates for web crawling.
\newblock In {\em Proceedings of the 16th International World Wide Web
  Conference}, 2007.

\bibitem{netlogic}
N.~Microsystems.
\newblock Netlogic microsystems: http://www.netlogicmicro.com/.

\bibitem{MP54}
M.~Minsky and S.~Papert.
\newblock Perceptrons.
\newblock In {\em MIT Press}, 1969.

\bibitem{MNP06}
R.~Motwani, A.~Naor, and R.~Panigrahi.
\newblock Lower bounds on locality sensitive hashing.
\newblock In {\em Proceedings of the 22nd Annual ACM Symposium on Computational
  Geometry}, 2006.

\bibitem{tcam:survey}
K.~Pagiamtzis and A.~Sheikholeslami.
\newblock Content-addressable memory (cam) circuits and architectures: A
  tutorial and survey.
\newblock {\em IEEE Journal of Solid-State Circuits}, 41(3):712--727, 2006.

\bibitem{P06}
R.~Panigrahi.
\newblock Entropy based nearest neighbor search in high dimensions.
\newblock In {\em Proceedings of the 17th Annual ACM-SIAM Symposium on Discrete
  Algorithms}, 2006.

\bibitem{PTW08}
R.~Panigrahy, K.~Talwar, and U.~Wieder.
\newblock A geometric approach to lower bounds for approximate near-neighbor
  search and partial match.
\newblock In {\em Proceedings of the 49th Annual IEEE Symposium on Foundations
  of Computer Science}, 2008.

\bibitem{ravi:NLP:clustering}
D.~Ravichandran, P.~Pantel, and E.~Hovy.
\newblock Using locality sensitive hash functions for high speed noun
  clustering.
\newblock In {\em Proceedings of the 43rd Annual Meeting of the Association for
  Computational Linguistics}, 2005.

\bibitem{RIJ79}
C.~Rijsbergen.
\newblock {\em Information Retrieval}.
\newblock Butterworth, 1979.

\bibitem{SH09}
R.~Salakhutdinov and G.~Hinton.
\newblock Semantic hashing.
\newblock {\em Int. J. Approx. Reasoning}, 50(7):969--978, 2009.

\bibitem{49}
D.~Shah and P.~Gupta.
\newblock Fast updates on ternary-cams for packet lookups and classification.
\newblock In {\em Proceedings of Hot Interconnects VIII}, 2000.

\bibitem{Tzane}
G.~Tzanetakis and P.~Cook.
\newblock {\em {\it MARSYAAS} A Framework for audio analysis}.
\newblock Cambridge University Press.

\bibitem{expt:space:partitioning}
R.~Weber, H.~Schek, and S.~Blott.
\newblock A quantititative analysis and performance study for similarity search
  methods in high dimensional spaces.
\newblock In {\em Proceedings of the Twenty-Fourth International Conference on
  Very Large Data Bases}, 1998.

\bibitem{WTF08:CVPR}
Y.~Weiss, A.~Torralba, and R.~Fergus.
\newblock Small codes and large image databases for recognition.
\newblock In {\em IEEE Computer Society Conference on Computer Vision and
  Pattern Recognition}, 2008.

\bibitem{WTF08:NIPS}
Y.~Weiss, A.~Torralba, and R.~Fergus.
\newblock Spectral hashing.
\newblock In {\em Proceedings of the Twenty-Second Annual Conference on Neural
  Information Processing Systems}, 2008.

\bibitem{Yao:cellprobe}
A.~Yao.
\newblock Should tables be sorted.
\newblock {\em J. Assoc. Comput. Mach}, 28.

\bibitem{tao}
Y.~Tao, K.~Yi, C.~Sheng, and P.~Kalnis.
\newblock Quality and Efficiency in High Dimensional Nearest Neighbor Search.
\newblock In {\em SIGMOD}, 2009. % Proceedings of the Thirty-Fifth SIGMOD International Conference on Management of Data

\end{thebibliography}
\end{document}